%% file: esakiaarxivfull.tex
\pdfoutput=1

\documentclass[toc=listof,toc=bibliography]{scrartcl}
\makeatletter
\DeclareOldFontCommand{\rm}{\normalfont\rmfamily}{\mathrm}
\DeclareOldFontCommand{\sf}{\normalfont\sffamily}{\mathsf}
\DeclareOldFontCommand{\tt}{\normalfont\ttfamily}{\mathtt}
\DeclareOldFontCommand{\bf}{\normalfont\bfseries}{\mathbf}
\DeclareOldFontCommand{\it}{\normalfont\itshape}{\mathit}
\DeclareOldFontCommand{\sl}{\normalfont\slshape}{\@nomath\sl}
\DeclareOldFontCommand{\sc}{\normalfont\scshape}{\@nomath\sc}
\makeatother

\usepackage{a4wide}
\input{esakiaarxiv_macros}

\usepackage[
    backend=biber,
    bibstyle=alphabetic,
    citestyle=alphabetic,
    sorting=nty,
    natbib=false,
    isbn=false,
    uniquename=false,
    firstinits=true,
    url=true, 
    doi=false,
    safeinputenc,
    bibencoding=auto,
    maxnames=5,
    maxalphanames=5
]{biblatex}
\AtEveryBibitem{%
  \clearlist{language}%
  \clearfield{month}%
  \clearlist{location}%
  \clearfield{pagetotal}%
}

\usepackage{listings}

\lstdefinelanguage{Coq}%
  {morekeywords={Variable,Inductive,CoInductive,Fixpoint,CoFixpoint,%
      Definition,Lemma,Theorem,Axiom,Local,Save,Grammar,Syntax,Intro,%
      Trivial,Qed,Intros,Symmetry,Simpl,Rewrite,Apply,Elim,Assumption,%
      Left,Cut,Case,Auto,Unfold,Exact,Right,Hypothesis,Pattern,Destruct,%
      Constructor,Defined,Fix,Record,Proof,Induction,Hints,Exists,let,in,%
      Parameter,Split,Red,Reflexivity,Transitivity,if,then,else,Opaque,%
      Transparent,Inversion,Absurd,Generalize,Mutual,Cases,of,end,Analyze,%
      AutoRewrite,Functional,Scheme,params,Refine,using,Discriminate,Try,%
      Require,Load,Import,Scope,Set,Open,Section,End,match,with,Ltac 
    },%
   sensitive, %
   morecomment=[n]{(*}{*)},%
   morestring=[d]",%
   literate={=>}{{$\Rightarrow$}}1 {>->}{{$\rightarrowtail$}}2{->}{{$\to$}}1
   {\/\\}{{$\wedge$}}1
   {|-}{{$\vdash$}}1
   {\\\/}{{$\vee$}}1
   {~}{{$\sim$}}1
  }

\title{Constructive Modalities with Provability Smack}
\subtitle{Author's Cut, v. 2.04: updated and extended electronic version}
\author{Tadeusz Litak} 
\titlehead{Informatik 8, FAU Erlangen-N\"{u}rnberg \newline\texttt{tadeusz.litak@fau.de}} 

 \dedication{\small\flushright \emph{To the memory of Leo Esakia and Dito Pataraia}}

\date{\small The references were last updated in 2015}

\begin{document}

\begin{titlepage}

\maketitle

\abstract{I overview the work of the Tbilisi school on intuitionistic modal logics of well-founded/scattered structures and its connections with contemporary theoretical computer science. 
Fixed-point theorems and their consequences are of particular interest. 
}


\tableofcontents

\paragraph*{Disclaimer}

\footnotesize

The paper is a modified version of \cite{Litak14:esakia}, which has
been commissioned for the highly recommended volume \cite{EsakiaOutstanding}. Please make clear
which version you are quoting, as even the numeration of all
environments is very different than in the version printed in the
volume (this is by no means the only difference).

\end{titlepage}


\section{Introduction}

The intended audience of \cite{EsakiaOutstanding} is probably aware that much of Leo Esakia's research concentrated on 
semantics for the intuitionistic logic $\iPC$, the modal logic $\cKW$  of L\"{o}b, its weakening $\cwGrz$ and intuitionistic-modal systems like the logic $\iKM$ or its weakening $\mHC$; see Table \ref{tab:addaxioms} for all definitions. $\cKW$ is also known as the  \textit{G\"{o}del-L\"{o}b} logic, but this name may suggest more personal involvement with the system than G\"{o}del ever had;  $\iKM$ or $\mHC$ will be discussed in 
 Section \ref{sec:deriv}. A central feature of semantics for such systems is  well-foundedness  or \emph{scatteredness}. 
While in the case of $\iPC$ well-foundedness is a sufficient, but not necessary condition---intuitionistic logic is complete wrt well-founded or even finite partial orders, but sound wrt much bigger class of structures---$\cKW$ and $\iKM$ require it even for soundness.  This is due to the fact that the latter two systems include a form of an explicit induction axiom: in the case of $\cKW$ the well-known L\"{o}b axiom (which here will be called the \emph{weak L\"{o}b axiom}) and in the case of $\iKM$ the \emph{strong L\"{o}b axiom}---less well-known to modal logicians, but as we are going to see, better  known to type theorists. Scatteredness is the topological generalization of well-foundedness; Simmons \cite{Simmons82:lnm}  provided tools necessary to define its point-free counterpart and the Tbilisi school 
noted that this notion also makes sense in the topos setting. In fact, the most generic way of defining scateredness is via modal syntax: as validity of the L\"{o}b principle for a  
 suitable ``later'' modality.

The interplay of relational, topological, point-free and algebraic aspects in the above paragraph should not feel unnatural to anybody familiar with Leo's attitude to research. 
Let us look at an important example how results can travel from one setting to another. 
 In the mid-1970's, it was established that L\"{o}b-like logics 
enjoy the so-called Fixpoint Theorem. 
At first, the intention was to grasp the algebraic content of 
 G\"{o}del's Diagonalization Lemma. 
 Yet in its own right it turned out to be one of the most fascinating results ever proved about such systems.  Section \ref{sec:fixtheorem} gives an overview of some of its applications and consequences. For now, let us just mention that 
Leo Esakia used it, e.g., to characterize  algebras for $\iKM$, 
 see Theorem \ref{th:algfixp} and Corollary \ref{cor:leofix} here. Furthermore, it seems to have inspired  the work on scattered toposes: \cite[Section 3]{EsakiaJP00:apal} claims to present its \emph{topos-theoretic counterpart}.  However, as the result central for the topos version (Theorem \ref{th:scanonexp} here) 
 does not even include modalities in its formulation, the word \textit{counterpart} has to be understood rather loosely. 

As we will see, in hindsight \cite{EsakiaJP00:apal} turns out to be closely connected to very recent developments in Theoretical Computer Science, in particular the work of Birkedal et al. on the \emph{topos of trees} \cite{BirkedalMSS12:lmcs}, itself an example of a scattered topos. Thus, it seems particularly regrettable that the spadework of the Tbilisi school has not been carried further and is not more widely known. 





The paper is structured as follows. 
Section \ref{sec:primer} recalls syntactic and semantic basics of \emph{intuitionistic} normal modal logics. 
Section \ref{sec:fixtheorem} focuses on fixpoint results for L\"{o}b-like systems. 
Section \ref{sec:deriv} introduces the work of the Georgian school 
on 
extensions of $\mHC$. 
Finally, Section \ref{sec:scatopoi} discusses scattered toposes, beginning with 
 an overview of the topos logic.

While the work is intended as an overview and claims to novelty are minimal, they are perhaps not entirely non-existent. Theorem \ref{th:algfixp} is the most general form of  \cite[Proposition 3]{Esakia06:jancl} I could think of and Section \ref{sec:complibt} reproves results on extensions of $\mHC$ 
using the framework of \cite{WolterZ97:al,WolterZ98:lw}; in fact, it seems that Corollary \ref{cor:blokesakia} 
is the first published proof of the corresponding extension of the Blok-Esakia Theorem announced in \cite{Esakia06:jancl}. 

\begin{myremark}\label{rem:coq}
As a part of a larger project, I formalized most of syntactic
derivations in the paper---in particular those relevant for Section
\ref{sec:scafix}--- in the Coq proof assistant. Readers interested in
this ongoing project are welcome to contact me. The formalization
covers, in particular, most of material in \cite{Esakia98:bsl}, which
prepared the proof-theoretical background for  \cite{EsakiaJP00:apal}
and inspired the title of this paper.   
\end{myremark}


\section{A Primer on Intuitionistic Modalities} \label{sec:primer}


Modal formulas over a supply of propositional variables  $\iS$ are defined by 
\[
A,B :\deq \bot \mid p \mid A \to B \mid A \wedge B \mid A \vee B \mid \ibox A
\]
where $p\in\iS$. The set is denoted by $\mF\iS$, but unless explicitly stated otherwise, I will keep $\iS$ fixed throughout and drop it from the notation. The purely intuitionistic language (i.e., without $\ibox$) will be denoted by $\iF$. Note that the syntax extended with a $\Diamond$ operator, intuitionistically \emph{not} definable from $\ibox$, is of no interest for us here. 

\def\tbskip{1.2mm}

\begin{table}

\scriptsize
\hrule
\vspace{1mm}

\caption{\label{tab:intaxioms}Axioms for 
$\iK$}

\vspace{5mm}



\begin{tabular}{@{}l@{\hspace{1mm}}l@{\hspace{25mm}}l@{\hspace{2mm}}l}
\multicolumn{4}{c}{Axioms of the intuitionistic propositional calculus, see, e.g., \cite[Sec. 1.3,(A1)-(A9)]{ChagrovZ97:ml}}\\[5mm] 
\multicolumn{4}{c}{Axiom for $\ibox$ \qquad} \\[\tbskip]
\multicolumn{2}{c}{\kax \qquad} & \multicolumn{2}{l}{$\ibox (A \to B) \to (\ibox A \to \ibox B)$} \\[4mm]
\multicolumn{2}{l}{Inference rule for $\iF$-fragment} & \multicolumn{2}{l}{Inference rule for modality} \\[\tbskip]
\mpr & \qquad $\inferrule{A \to B, \quad A}{B}$  & \necr & $\inferrule{A}{\ibox A}$ \\[\tbskip]
\end{tabular}
\vspace{1mm}
\hrule
\end{table}

 $\Gamma \subseteq \mF$ is \emph{a normal $\mF$-logic} or an \emph{intuitionistic normal modal logic} if it is closed under rules and axioms from Table \ref{tab:intaxioms} plus the rule of substitution. 
For any $\Gamma,\Delta \subseteq \mF$, $\Gamma \lpl \Delta$ will denote the closure of $\Gamma \cup \Delta$ under  substitution 
and the rules $\mpr$  and $\necr$. In the case of $\Delta = \{\alpha\}$, I will also write $\Gamma \lpl \alpha$. Occasionally, I will write $\Gamma \dpl \Delta$ for the closure under  substitution 
and $\mpr$, but without $\necr$. This notation is analogous to the one used in \cite{ChagrovZ97:ml}. 

\begin{table}
\footnotesize
\hrule
\vspace{1mm}

\caption{\label{tab:addaxioms}$\mF$ axioms and logics. See, e.g., \cite{Sotirov84:ml,WolterZ97:al,WolterZ98:lw} for more (also in the syntax extended with a $\Diamond$ operator). $\ribox A$ below abbreviates $A \wedge \ibox A$}
\vspace{1mm}
\begin{tabular}{>{$}l<{$}>{$}l<{$}@{\hspace{4mm}}>{$}l<{$}>{$}l<{$}
} 
\clax &  ((B \to A) \to B) \to B & \claxb & A \vee \neg A \\
& \iCl = \iPC \dpl \clax = \iPC \dpl \claxb && \\[\tbskip]
\kax & \ibox (A \to B) \to (\ibox A \to \ibox B) 
& \kaxb & \ibox (A \wedge B) \leftrightarrow (\ibox A \wedge \ibox B) \\
& \iK = \iPC \lpl \kax = \iPC \lpl \kaxb && \cK = \iK \lpl \iCl  
\\[\tbskip] 
\kfax & \ibox A \to \ibox\ibox A & & \\
& \iKFour = \iK \lpl \kfax && \cKFour = \iKFour \lpl \iCl
\\[\tbskip] 
\cfax &  \ibox\ibox A \to \ibox A & 
& \iCFour = \iK \lpl \cfax
\\[\tbskip]
\rax &  A \to \ibox A &   
\raxb & (A \to B) \to (\ibox A \to \ibox B) \\ 
& \iR = \iK \lpl \rax = \iK \lpl \raxb & \multicolumn{2}{l}{Note that above $\iR$, $\lpl$ is the same as $\dpl$} \\
&  \multicolumn{2}{l}{In using the symbol $\axst{R}$, I follow \cite{FairtloughM97:ic}}\\[\tbskip]
\tax & \ibox A \to A && \iSFour = \iT \lpl \iKFour\\
& \iT = \iK \lpl \tax && \iITriv = \iT \dpl \iR
\\[\tbskip]
\laxax & 
(A \vee \ibox\ibox A) \to \ibox A & & \\
& \iLax = \iK \lpl \laxax = \iCFour \lpl \iR && \\[\tbskip]
\kwax & \ibox (\ibox A \to A) \to \ibox A & 
  \henkax & \ribox (A \leftrightarrow \ibox A) \to A \\
\ufpax & \multicolumn{3}{l}{$\ribox (B \leftrightarrow A[B/p]) \to (\ribox (C \leftrightarrow A[C/p]) \to (B \leftrightarrow C))$} \\
& \multicolumn{3}{l}{$\iKW = \iK \lpl \kwax = \iKFour \lpl \henkax = \iKFour \lpl \ufpax$} \\
& \text{(see Theorem \ref{th:lobax} below)} && \cKW = \iKW \lpl \iCl
\\[\tbskip]
\rlax &  (\ibox A \to A) \to A & 
 \gtax & (\ibox A \to A) \to \ibox A  \\ 
& \multicolumn{3}{l}{$\iRLob = \iK \lpl \rlax = \iK \lpl \gtax = \iKW \dpl \iR$} \\
& \multicolumn{3}{l}{The form $\gtax$ comes from Goldblatt \cite{Goldblatt81:mlq}} 
\\[\tbskip]
\grzax & \ibox(\ibox(A \to \ibox A) \to A) \to \ibox A & \sgrzax & \ibox(\ibox(A \to \ibox A) \to A) \to A  \\
& \cwGrz = \cKFour \lpl \grzax & & \csGrz = \cK \lpl \sgrzax \\
& \multicolumn{3}{l}{Note we only consider here classical variants of $\grzax$} \\[\tbskip] 
\nxtax &  \ibox A \to (((B \to A) \to B) \to B) &  
\nxtaxb & \ibox A \to ((B \to A) \vee B)  \\
& \multicolumn{2}{l}{$\iNext = \iK \lpl \nxtax = \iK \lpl  \nxtaxb$}  
\\
& \multicolumn{2}{l}{$\iNext$ stands for \emph{Cantor-Bendixson}, see Sec. \ref{sec:deriv}}
& \imHC = \iR \dpl \iNext \\ 
& \iNKW = \iNext \lpl \iKW &&
\iNLob = \iNext \lpl \iLob \\
& \multicolumn{3}{l}{$\iKM = \iRLob \dpl \iNext = \iLob \dpl \imHC$}
\\[\tbskip]
\gdax & (A \to B) \vee (B \to A)  &    &  \iLC = \iPC \dpl \gdax \\
\dtax & \ibox (\ribox A \to B) \vee \ibox(\ribox B \to A) &&  　\\[\tbskip]
\verax & \ibox A & \veraxb & \ibox \bot\\
& \multicolumn{2}{c}{$\iIVer = \iK \lpl \verax = \iK \lpl \veraxb$} &
\\[\tbskip]
\ndax & \neg\neg\ibox\bot & \dax & \neg\ibox\bot \\
& \iND = \iK \lpl \ndax && \iD = \iK \lpl \dax  
\end{tabular}

\hrule
\end{table}

\input{intKfig}

$\iK$ is the smallest intuitionistic normal modal logic, i.e., $\iPC \lpl \kax$.  $\iPC$---the intuitionistic propositional calculus---can be thus defined as the intersection of $\iK$ and $\iF$. Table \ref{tab:addaxioms}  provides a list of additional axioms and logics which will be of interest to us. 
 $\iKW$, $\iRLob$, $\mHC$ and $\iKM$ are of particular importance. As
 we see in Table \ref{tab:addaxioms} and  Figure \ref{fig:intKfig}, there are several ways in which these and related systems can be axiomatized.  
 In particular, 
 we have 

\begin{thm}[Ursini \cite{Ursini79:sl}, following Smorynski for the classical case] \label{th:lobax}
The following formalisms have the same set of theorems:
\begin{enumerate}
\item $\iKW$ as defined in Table \ref{tab:addaxioms}
\item $\iKFour \lpl \inferrule{\ibox A \to A }{A}$
\item $\iKFour \lpl \inferrule{\ibox A \to A }{\ibox A}$
\item $\iKFour \lpl \ufpax =\ribox (B \leftrightarrow A[B/p]) \to (\ribox (C \leftrightarrow A[C/p]) \to (B \leftrightarrow C))$
\item $\iKFour \lpl \henkax = \ribox (A \leftrightarrow \ibox A) \to A$
\end{enumerate}
\end{thm}

A variable $p \in \iS$ is \emph{$\ibox$-guarded} in $A \in \mF$ if all its occurrences are within the scope of $\ibox$. This notion will be used repeatedly in connection with $\iKW$ and its extensions.

\begin{myremark}\label{rem:deriv}
Equalities in Table \ref{tab:addaxioms} and joins in Figure \ref{fig:intKfig} should be in fact treated as a large lemma on interderivability 
for a number of intuitionistic normal modal axioms. Of particular interest for this chapter are: the derivability of $\rax$  from $\rlax$, which mirrors derivability of $\kfax$ from $\kwax$; interderivability between $\gtax$ and $\rlax$; equivalence between either of these and the conjunction of $\kwax$ and $\rax$; two different ways of axiomatizing $\iNext$ by using $\nxtax$ and $\nxtaxb$. All these statements are made assuming $\iK$.
\end{myremark}

All normal 
$\fullsig$-logics as defined above---more precisely, their associated global consequence relations---are strongly finitely algebraizable; see standard references like
\cite{BlokP89:ams,Font06:sl,FontJP03a:sl,FontJP09:sl,Rasiowa74:aatnl}   for basic notions of algebraic logic and a more detailed discussion. Strong finite algebraizability also applies to normal logics in any fragment of $\fullsig$ containing $\to$. This is due to the fact that  all these systems are \emph{implicative logics} in the sense of Rasiowa \cite{Rasiowa74:aatnl}.  Given a normal logic $\Gamma$, I will call the corresponding class of algebras obtained via the algebraization process  \emph{$\Gamma$-algebras}, e.g., $\iK$-algebras, $\iKW$-algebras, $\cwGrz$-algebras etc. $\iK$-algebras are obviously special cases of \HAO s---Heyting algebras with operators---namely \HAO s with a single unary operator. Recall that an \emph{operator} on a Heyting algebra is an operation preserving $\top$ and finite meets. 
An operator on a Heyting algebra which turns it into a $\Gamma$-algebra will be called a \emph{$\Gamma$-operator}. 

Finally, recall that  for any algebra $\gH$, a \emph{$\gH$-polynomial} is a term in the similarity type of $\gH$ \emph{enriched with a separate constant for each element of $\gH$} \cite[Definition 13.3]{BurrisS81:cua}. In my notation for polynomials, I will not distinguish between an element of $\gH$ and its corresponding constant. Moreover, I will use elements of $\iS$ (i.e., propositional variables) as variables of polynomials, consistent with the general policy of blurring the distinction between logical formulas and algebraic terms. The notion of $\ibox$-guardedness for polynomials will be used in the same way as for formulas.


\subsection{Relational Semantics} \label{sec:relsem}

It is pretty obvious what should be the ingredients of a Kripke frame (a relational structure) for intuitionistic normal modal logic: it should be of the form $(W,\icc,\mcc)$, where $\icc$ is a poset order used to interpret intuitionistic connectives, whereas $\mcc$ is the modal accessibility relation used to interpret $\ibox$. Valuations of propositional variables should be defined as as $V: \iS \to \upset{W}$ and the inductive extension to intuitionistic connectives using $\icc$ is standard. However, once we set to provide a precise interpretation of modalities, there are choices to be made. Here is what \cite[Section 3.3]{Simpson94:phd}---one of most comprehensive overviews---says on the subject: 

\begin{quote}
One could take also
the usual satisfaction clauses for the modalities in modal models \dots 
However, an essential feature of intuitionistic models is the monotonicity lemma \dots
If the standard satisfaction clauses for the modalities are used
then the monotonicity lemma does not hold. There are two possible remedies.
\acls One is to modify the satisfaction clauses. This might be a reasonable thing to do,
for one might wish to use the partial order to give a more intuitionistic reading of
the modalities. \bcls The other remedy is to impose conditions on models that ensure
that the monotonicity lemma does hold.
\end{quote}

The \emph{monotonicity lemma} referred to in the above is of course the fact that the denotation of any intuitionistic formula in any Kripke model is upward closed. The two solutions lead to two alternative satisfaction definitions:

\begin{itemize}
\item $\ibox_\acls A \deq  \{ w \in W \mid \forall v,x.(w \icc v \text{ and } v \mcc x \text{ implies } x \in A) \}$ 
\item $\ibox_\bcls A \deq  \{ w \in W \mid \forall x.(w \mcc x \text{ implies } x \in A) \}$ 
\end{itemize}

By definition, $\ibox_\acls A \in \upset{W}$ for any $A \subseteq W$. This, however, is easily seen to fail in general for $\bcls$, even for some $A \in \upset{W}$. This is precisely while we need to impose additional conditions on the interaction of $\icc$ and $\mcc$:

\begin{lem}[\cite{BozicD84:sl}] \label{lem:mincond}
For any $(W,\icc,\mcc)$, the condition 
\begin{center}
for any $A \in \upset{W}$, $\ibox_\bcls A \in \upset{W}$ 
\end{center}
(i.e., $\upset{W}$ is a subalgebra of $\powset{W}$ also with respect to $\ibox_\bcls$) is satisfied iff the following condition holds:
\begin{equation} \label{eq:mincond}
\icc \cmp \mcc \subseteq \mcc \cmp \icc
\end{equation}
\end{lem}


Curiously, condition \refeq{eq:mincond}, even though isolated in \cite{BozicD84:sl} and discussed in references such as \cite{Simpson94:phd}, does not seem to appear too often in intuitionistic modal literature. Most references, even those which adopt interpretation \bcls, assume conditions stronger than \refeq{eq:mincond}. Goldblatt \cite{Goldblatt81:mlq} requires 
\begin{equation} \label{eq:goldblattcond}
\icc \cmp \mcc \subseteq \mcc
\end{equation}
but in most of intuitionistic modal logic literature (see, e.g., \cite{Sotirov84:ml,WolterZ97:al,WolterZ98:lw}) we have still stronger
\begin{equation} \label{eq:wzcond}
\icc \cmp \mcc \cmp \icc = \mcc
\end{equation}

Obviously, \refeq{eq:wzcond} implies \refeq{eq:goldblattcond} and \refeq{eq:goldblattcond} implies \refeq{eq:mincond}, as $\icc$ is a partial order. But as noted in \cite{BozicD84:sl,Goldblatt81:mlq}
\begin{equation} \label{eq:upcond}
\mcc \cmp \icc = \mcc
\end{equation}

(i.e., for any $x$, $\mcc^{-1}[x]$ is upward closed) cannot be forced by any axiom, as any frame $(W,\icc,\mcc)$ can be shown to be semantically equivalent to $(W,\icc,\mcc')$ satisfying \refeq{eq:upcond} by putting $\mcc' = \mcc \cmp \icc$. 

In presence of \refeq{eq:upcond}, \refeq{eq:mincond} implies \refeq{eq:wzcond} and hence all these conditions become equivalent. Thus, we can follow the others and assume \refeq{eq:wzcond}, but from a logical point view, only \refeq{eq:mincond} would be necessary.

\begin{myremark}\label{rem:wzcond}
A justification for adopting the strongest condition \refeq{eq:wzcond}
is provided by completeness/canonicity result, for which the earliest
references seem to be \cite{BozicD84:sl,Sotirov84:ml}: when we define
$\mcc$ between prime filters of a Heyting algebra with a normal
$\ibox$ in a standard way, the resulting Kripke frame will satisfy
\refeq{eq:wzcond}. This provides a justification of
\refeq{eq:wzcond} and a bridge between algebraic and Kripke semantics,
i.e., a completeness result for $\iK$ wrt Kripke frames satisfying \refeq{eq:wzcond} with $\ibox$ read as dictated by \bcls.
\end{myremark}

\begin{myremark}
What is the connection with \acls-semantics? It is easy to see that $\acls$-reading of $\ibox$ wrt $\mcc$ yields the same operator on $\upset{W}$ as $\bcls$-reading wrt $\mcc$ modified to satisfy \refeq{eq:goldblattcond}: that is, the closure of $\mcc$ under pre-composition with $\icc$. Note here that in order to obtain a relation satisfying \refeq{eq:wzcond}, one needs to close $\mcc$ under both pre- and post-composition with $\icc$.
\end{myremark}

Thus, as stated in \cite{Simpson94:phd}

\begin{quote}
it turns out that
both semantics induce the same intuitionistic modal logic. \dots 
Only when the $\Diamond$ connective is added do the differences in the semantics become apparent.
\end{quote}

The reason why \acls is often seen in the literature (see
\cite{Simpson94:phd} and references therein and also, e.g,
\cite{AlechinaMdPR01:csl,PlotkinS86:tark}) is that it becomes pretty
natural when intuitionistic modal logic is considered as a fragment of
intuitionistic predicate logic---and, in close connection, when one
wants to interpret $\Diamond$ using the same relation as for
$\ibox$. However, this is not what we are interested in here. Thus, in
what follows, I stick to \bcls reading of modality (henceforth $\bcls$
is dropped as a subscript) and impose condition \refeq{eq:wzcond} on all
frames. This is justified not only by Remark \ref{rem:wzcond} above,
but also by the fact that intuitionistic modal semantics in
practically all references relevant for further discussion satisfied
\refeq{eq:wzcond} anyway. Let us summarize the discussion above with the
following definition:

\emph{A Kripke frame} or \emph{a relational structure} 
is of the form $(W,\icc,\mcc)$, where
\begin{itemize}
\item $\icc$ is a poset order used to interpret intuitionistic connectives
\item $\mcc$ is the modal accessibility relation used to interpret $\ibox$ and
\item $\icc \cmp \mcc \subseteq \mcc$, where $\cmp$  is relation composition
\end{itemize}

A \emph{valuation} 
 is a mapping $V: \iS \to \upset{W}$, where $ \upset{W}$ is the Heyting algebra of upward closed sets of $W$ and the inductive extension to $\mF\iS$ 
 is standard.


 \begin{myremark}\label{rem:brouwerian}
 When at least one of lattice connectives is removed, the situation at first sight appears more complicated. While the papers proving the separation property of $\iPC$ \cite{Horn62:jsl,Milne10:rsl,Prawitz65:nd} showed that its reducts 
  remain complete wrt relational semantics, a Stone-type representation theorem for arbitrary algebras 
 would seem more problematic; the one for Heyting algebras 
 relies crucially on the fact that they have distributive lattice reducts. However, a series of papers beginning with \cite{Kohler81:ams} and finishing with \cite{BezhanishviliJ12:acs} established  that Brouwerian semilattices enjoy in fact Stone-, Priestley- and Esakia-type dualities. 
\end{myremark}

\begin{myremark}
It is worth discussing whether we can reconstruct this relational
semantics from coalgebraic point of view, just like in the classical
case. The relationship between the category of Heyting algebras (\HA)
and the category of posets with bounded morphisms (say, \PABM) is the
same as in the case of boolean algebras and sets (dual adjunction,
restricting to equivalence in the finite case). Kripke frames for
classical modal logic can be thought as coalgebras for covariant
powerset functor. The counterpart of that functor in \PABM is the
covariant upset functor $\upset{\cdot}$, which sends $(W,\icc)$ to
$(\upset{W},\supseteq)$ and $\upset{f}(A) = f[A]$. Then if we attempt
to define relational structures as \PABM-morphisms, i.e., bounded
morphisms from  $(W,\icc)$ to $(\upset{W},\supseteq)$, the very choice
of codomain of this morphism forces \refeq{eq:upcond} and the forth
condition of bounded morphism becomes precisely
\refeq{eq:goldblattcond}, so altogether we get just
\refeq{eq:wzcond}. However, the back condition seems to impose a
rather special condition on $\mcc$: 
 for any $x \in W$ and any upward closed $A \subseteq
\mcc^{-1}[x]$, there is $y$ s.t. $x \icc y$ and $\mcc^{-1}[y] =
A$. There are many natural examples of structures satisfying
\refeq{eq:wzcond} where this condition would fail and in general I do
not know how many logics would be possibly complete with respect to
frames of this kind. It is not necessarily harmless for some of the
axioms considered here. However, when we can forget about the back condition for morphisms (i.e., in case of modalities over distributive lattices rather than over Heyting algebras), \refeq{eq:wzcond} emerges as the natural basic semantic condition also from the coalgebraic point of view.
\end{myremark}

Table \ref{tab:semcond} lists semantic conditions corresponding to  modal axioms. For $\iKW$ in particular, we have: 

\begin{thm}[\cite{Ursini79:sl}] \label{th:ursini}
A structure $(W,\icc,\mcc)$ 
validates $\iKW$ iff 
\begin{itemize}
\item $\mcc$ is transitive, i.e., $\mcc \cmp \mcc \, \subseteq \, \mcc$ 
and
\item $\mcc$ is \emph{$\upset{W}$-Noetherian}: for any $A \in \upset{W}$, if $A \neq W$, then there is $w \in \ibox A - A$, i.e., a $\ibox$-maximal non-$A$ point  
\end{itemize}
\end{thm}

\begin{proof}[Proof (sketch)]
We show only the ''if'' direction.  Assuming $\kwax$ fails under a valuation $V$, take $B$ to be the extension of $\kwax$ under $V$ and 
show  that $\ibox B \subseteq B$. This means that $W-B$ witnesses the failure of 
Noetherianity.
\end{proof}

Note that Theorem \ref{th:ursini} together with the $\rax$ row
of Table \ref{tab:semcond}  provides also the condition
for $\iSL$: as already stated in Table\ref{tab:addaxioms} and Figure
  \ref{fig:intKfig}, $\iSL$ is obtained as $\iKW \lpl \rax$, so one
  simply takes the conjunction of the two conditions.

\begin{table}
\footnotesize
\hrule

\caption{\label{tab:semcond} Semantic counterparts of intuitionistic modal axioms. See, e.g., \cite{Dosen85:sl,Sotirov84:ml,WolterZ97:al,WolterZ98:lw} for more. $\diag \deq \{(w,w) \mid w \in W \}$ and $\iccs \deq \icc - \diag$, i.e., it is the strict version of the intuitionistic poset order. For $R \in \{\mcc,\icc\}$, $V\subseteq W$, $\reluc{V}{R} \deq \{w \in W \mid \exists v \in V.vRw\}$ and 
$\reldc{V}{R} \deq \{w \in W \mid \exists v \in V.wRv\}$}
\vspace{1mm}
\begin{tabular}{>{$}l<{$}@{\hspace{0.2cm}}>{$}l<{$}@{\hspace{0.3cm}}>{$}l<{$}@{\hspace{0.2cm}}>{$}l<{$}@{\hspace{0.3cm}}>{$}l<{$}@{\hspace{0.2cm}}>{$}l<{$}
} 
\text{Axiom} & \text{Semantic condition} & \text{Axiom} & \text{Semantic condition} & \text{Axiom} & \text{Semantic condition} 
\\[\tbskip] 
\clax & \icc =  \diag&
\kfax &  \mcc \cmp \mcc \, \subseteq \, \mcc 
&
\cfax  & \mcc  \, \subseteq \, \mcc \cmp \mcc \\[\tbskip] 
\rax   
& \mcc \, \subseteq \, \icc  
&
\tax &  \icc \, \subseteq \, \mcc  &
\kwax & \text{see Theorem \ref{th:ursini}} 
 \\[\tbskip] 
\nxtax & \iccs \, \subseteq \, \mcc & \verax & \mcc = \emptyset & \dax & \reldc{W}{\mcc} = W \\[\tbskip] \ndax & \reldc{(W-\reldc{W}{\mcc})}{\icc} = W &&&&
\\[\tbskip]
\end{tabular}

\hrule

\end{table}

\cite{Ursini79:sl} provides an interesting motivation for this semantics of $\iKW$ in terms of \emph{projects} and \emph{streamlines} in a research. It also provides other important results, such as the finite model property and decidability. 




\section{The Fixpoint Theorem} \label{sec:fixtheorem}

Theorem \ref{th:lobax} above gave us several equivalent axiomatizations of $\iKW$. 
In particular, $\ufpax$  forces \emph{uniqueness of fixed points}. 
\begin{defn}
Let $B$ be a formula of $\mF$ (its denotation in a given algebra and under a given valuation). $B$ is a \emph{fixed point} of (the term function associated with) $A \in \mF$ relative to $p$ in a given normal logic $\Gamma$ (here always an extension of $\iKW$) if 
$B \leftrightarrow A[B/p] \in \Gamma$ and 
$p$ does not occur in $B$. 
\end{defn}

According to \cite{Smorynski79:sl}, 
the fact that $\ufpax$ holds in $\cKW$ was discovered independently by Bernardi, Sambin and de Jongh. We know thus that, surprisingly, in $\iKW$ a syntactic fixed point of an expression is unique up to equivalence whenever it exists; same applies to all of its extensions, such as $\iLob$ or $\iKM$. 
But do they exist at all? An even more amazing fact is that they not only do exist---under the assumption of $\ibox$-guardedness on $p$---but are effectively computable. This is guaranteed by the following algebraic (or propositional, if one prefers) variant of G\"{o}del's Diagonalization Lemma.  Sambin \cite{Sambin76:sl} proved it for $\iKW$ itself and de Jongh proved it for $\cKW$ building on an earlier result by Smorynski, another proof being found soon afterwards by Boolos; the reader is referred to \cite{Boolos93:lop,BoolosS91:sl,Muravitsky:here,Smorynski79:sl} for more on its history and the connection with G\"{o}del's result:

\begin{thm}[Diagonalization]
\label{theorem:diagonalization}
For any $A$ and $p$, there exists a constructively obtained formula $\diagf{A}{p}$ s.t. 
\begin{enumerate}
\item $\diagf{A}{p} \leftrightarrow B \in \iKW$, where $B$ is obtained from $A$ by replacing all occurrences of $p$ under $\ibox$ by $\diagf{A}{p}$
\item $A$ and $\diagf{A}{p}$ have provably the same fixed points with respect to $p$, that is, for any $C$ not containing $p$ we have 
$$\ribox(C \leftrightarrow A[C/p]) \leftrightarrow \ribox(C \leftrightarrow \diagf{A}{p}[C/p]) \in \iKW$$
\end{enumerate}
\end{thm}
 
Clearly, if $p$ is $\ibox$-guarded in $A$, then
 $B$ in the first clause is precisely $A[\diagf{A}{p}/p]$ and 
 $\diagf{A}{p}$ does not contain $p$, hence being trivially its own fixed point wrt $p$. 
Thus, in such a situation $\diagf{A}{p}$ itself is also the unique fixed-point of $A$ with respect to $p$! 

\begin{proof}[Proof of Theorem \ref{theorem:diagonalization}, sketch]
We only give a sketch of how $\diagf{A}{p}$ is constructed. Any formula $A(p,\vars{q}) \in \mF$ in variables $p,\vars{q} \in \iS$ can be written as $B(\ibox C_1(p,\vars{q}), \dots,\ibox C_k(p,\vars{q}),p,\vars{q})$, where $B \in \iF$ (i.e., is a formula without $\ibox$) and $\overline{C} \in \mF$. Clearly, if $k=0$, then $A$ itself belongs to $\iF$ and, in particular, there are no occurrences of $p$ under $\ibox$. Hence we can take $\diagf{A}{p}$ to be $A$ itself. Otherwise, the proof can be conducted by induction on $k$, as we already have the base step. 
For any $i \leq k$, set 
$$A_i \deq B(\ibox C_1(p,\vars{q}), \dots, \ibox C_{i-1}(p,\vars{q}), \top, \ibox C_{i+1}(p,\vars{q}), \dots, \ibox C_k(p,\vars{q}),p,\vars{q}).$$
By definition, the inductive hypothesis applies to $A_i$. Now we set
\[
\hspace{2cm} \diagf{A}{p} \deq B(\ibox C_1(\diagf{A_1}{p}/p,\vars{q}), \dots, \ibox C_k(\diagf{A_k}{p}/p,\vars{q}),p,\vars{q}).
\hspace{1.8cm} 
\]
\end{proof}

\begin{myremark}\label{rem:kinderg}
In fact, extensions of $\iRLob$---in particular $\iKM$---allow a much simpler proof of Theorem \ref{theorem:diagonalization}  and a much simpler algorithm for computing these fixpoints: it is enough to substitute $\top$ for $p$. This follows already from observations made by Smorynski in \cite[Lemma 2.3]{Smorynski79:sl} and has been discussed explicitly in \cite[Propositions 4.2--4.6]{deJonghV95:emb}. De Jongh and Visser describe $\iRLob$ as \emph{a kind of Kindergarten Theory in which all the well-known syntactical results
of Provability Logic have extremely simple versions}. 
\end{myremark}


\begin{myremark}\label{}
It is known that at least in the case of $\cKW$ a non-constructive and non-explict form of Theorem \ref{theorem:diagonalization}  can be obtained already from uniqueness of fixed-points combined with the Beth definability theorem for $\cKW$, see, e.g.,  \cite{Boolos93:lop,Smorynski79:sl} for more information. However, as should be clear from the discussion below, the very fact that fixed points are obtained explicitly and constructively seems an advantage not to be given up lightly.
\end{myremark}

Theorem \ref{theorem:diagonalization} has a nice algebraic corollary. I present it here as a  more general version of \cite[Proposition 3]{Esakia06:jancl}.

\begin{thm} \label{th:algfixp}
A $\iKFour$-algebra $\gH$ is a $\iKW$-algebra iff every $\gH$-polynomial $t(p)$ in one $\ibox$-guarded variable $p \in \iS$ 
 has a fixed point. 
\end{thm}

\newcommand{\qhm}{\iml_{\hml}}

\begin{proof}
\emph{The ``only if'' direction.} This is a direct corollary of Theorem \ref{theorem:diagonalization}.

\emph{The ``if'' direction.} Given any $\hml \in \gH$, consider the polynomial $t(p) = \ibox p \to \hml$. As $p$ is $\ibox$-guarded in it, it has a fixed point $\qhm \in \gH$; that is, $\qhm = \ibox \qhm \to \hml$. By the fact that $\to$ is conjugate (or residual) to $\wedge$, one half of this equality is equivalent to $\qhm\wedge\ibox\qhm \leq \hml$. On the other hand, $\hml \leq \ibox \qhm \to \hml = \qhm$ by general implication laws. Taken together, these two inequalities imply $\qhm \wedge \ibox\qhm = \hml \wedge\ibox\hml$: the $\leq$ direction from the first inequality, normality and $\kfax$, the $\geq$ direction from the second inequality and monotonicity of $\ibox$. Using normality again, we get $\ibox\qhm \wedge \ibox\ibox\qhm = \ibox\hml \wedge\ibox\ibox\hml$ and using $\kfax$ again, we arrive at $\ibox\qhm = \ibox\hml$. Then we get $\ibox(\ibox\hml \to \hml) = \ibox(\ibox\qhm \to \hml) = \ibox\qhm = \ibox\hml$. As $\hml \in \gH$ was chosen arbitrarily, we have that $\gH$ is a $\iKW$-algebra.
\end{proof}

There is an analogy between the above result and alternative axiomatizations for $\iKW$ presented in \cite{Ursini79:sl}.

\begin{cor} \label{cor:leofix} 
\mbox{\newline}
\begin{itemize}
\item A $\iR$-algebra $\gH$ is a $\iRLob$-algebra iff every $\gH$-polynomial $t(p)$ in one $\ibox$-guarded variable $p \in \iS$ 
 has a fixed point. 
\item \cite[Proposition 3]{Esakia06:jancl} A $\mHC$-algebra $\gH$ is a $\KM$-algebra iff every $\gH$-polynomial $t(p)$ in one $\ibox$-guarded variable $p \in \iS$ 
 has a fixed point. 
\end{itemize}
\end{cor}

 As stated above, 
Theorem \ref{theorem:diagonalization} occurred first as an algebraization of G\"{o}del's Diagonalization Lemma. 
While the connection between $\iKW$ and Heyting Arithmetic \HA\, is not as tight as the one between $\cKW$ and Peano Arithmetic \PA\, established by the completeness result of Solovay \cite{Solovay76:ijm} (see also \cite[Chapter 9]{Boolos93:lop}), \cite{Sambin76:sl} notes that Theorem \ref{theorem:diagonalization} yields a counterpart of the Diagonalization Lemma for  \emph{any intuitionistic first-order theory with a canonical derivability predicate}, including obviously \HA. At any rate, the relevance of fixpoint results for L\"{o}b-like logics is not limited to arithmetic. 

\begin{myremark}\label{}
It is worth mentioning here that---unlike the case of $\PA$---the search for a complete axiomatization of the provability logic of \HA\, is 
not over yet; \cite{Iemhoff01:phd} gives a fascinating account. Regarding the arithmetical interpretation of $\iSL$, see the discussion of $\HA^*$ in \cite[Sec. 4--5]{deJonghV95:emb}.
\end{myremark}


To begin with, in the classical case one can use Theorem \ref{theorem:diagonalization} to show that explicit smallest or greatest fixed-point operators are eliminable over $\cKW$. In other words, adding  $\mu$ or $\nu$ does not increase the expressive power of the classical modal logic of transitive and conversely well-founded structures; see \cite{AlberucciF09:sl,Benthem06:sl,Visser05:lncs}. Note that this includes all correctly formed expressions with $\mu$, without assuming that all occurrences of $p$ are $\ibox$-guarded: as usual, they only have to be positive. \cite[Section 3]{CateFL10:jancl} discusses an application in the context of expressivity of navigational fragments of XML query languages.

While I am not aware of an exact analogue of the results in
\cite{AlberucciF09:sl,Benthem06:sl,Visser05:lncs} in the
intuitionistic context,\footnote{It is a theme I am presently
  collaborating on with Albert Visser.} L\"{o}b-like modalities---more specifically, variants of systems $\iKW$ and $\iRLob$---have recently become rather popular  in type theory. Examples include:
\begin{itemize}
\item \emph{modality for recursion} \cite{Nakano00:lics,Nakano01:tacs} 
\item \emph{approximation modality} \cite{AppelMRV07:popl} 
\item \emph{guardedness type constructor} \cite{AtkeyMB13:icfp} 
\item \emph{next-step modality/next clock tick} \cite{KrishnaswamiB11:icfp,KrishnaswamiB11:lics} 
\item \emph{later operator} \cite{BentonT09:tldi,BirkedalM13:lics,BirkedalMSS12:lmcs,JaberTS12:lics}
\end{itemize}

One of reasons is precisely that such modalities 
 guarantee existence and uniqueness of fixed-points of suitably guarded type expressions. 
However, the modal spadework of 1970's seems rarely acknowledged.  
In \cite{Nakano00:lics}, which may be credited with introducing intuitionistic L\"{o}b-like modalities to the attention of this community, we find the following statement:

\begin{quote}
Similar results concerning the existence of fixed points
of proper type expressions \dots \emph{could
historically go back to the fixed point theorem of the logic of
provability} \dots \emph{The difference is that our logic is
intuitionistic}, and fixed points are treated as sets of realizers [the emphasis is mine--T.L.].
\end{quote}

This formulation suggests that Nakano was not aware  that the \emph{intuitionistic} fixed-point theorem had been already proved  in \cite{Sambin76:sl}, not to mention improvements possible above $\iRLob$  (cf. Remark \ref{rem:kinderg}). The only related  references quoted in \cite{Nakano00:lics} focus on classical $\cKW$---e.g, \cite{Boolos93:lop}---and in later papers even these are omitted. A valuable part of the logical tradition seems lost this way. Let us see what insights can be found in the work of the Tbilisi school.  

\section{Leo Esakia and Extensions of $\mHC$} 
\label{sec:deriv}

\subsection{$\mHC$ and Topological Derivative}

Leo Esakia and collaborators 
devoted special attention to the system $\imHC$ and its extensions. \cite{Esakia06:jancl} is an excellent overview. The abbreviation $\imHC$ stood for \emph{modalized Heyting calculus}. The reader may find the name surprising; after all, 
many natural intuitionistic modal systems are not subsystems of $\imHC$. Esakia \cite{Esakia06:jancl} was perfectly aware of that:

\begin{quote}
The postulate $\rax$ is not typical, while the postulate $\nxtaxb$ stresses even more ``nonstandardness'' of the chosen basic system $\mHC$ and of its extension $\KM$, which
enables one to draw a conventional ``demarcation line'' between $\mHC$ and the standard
intuitionistic modal logics.  
\end{quote}

\begin{myremark}\label{rem:monads}
Both axioms seem ``nonstandard'' mostly if one focuses on these intuitionistic modal logics which are obtained from  popular classical systems. It is enough to look at Table \ref{tab:semcond} to realize why it must be so: $\nxtaxb$ is trivially derivable in $\cK$ (its consequent being a classical tautology), while the combination of $\clax$ and $\rax$ yields that $\mcc \subseteq \diag$. That is, the only classical frames for $\cR$ are disjoint unions of reflexive and irreflexive points. However, $\rax$ is nowhere as pathological in a properly intuitionistic setting. There are many references on systems different from $\mHC$ and $\iKM$ where nevertheless $\rax$ is still derivable or even explicitly included as an axiom. A short and inexhaustive list includes \cite{Abadi07:gdp,Aczel01:mscs,BentonBP98:jfp,BirkedalMSS12:lmcs,Curry52:jsl,FairtloughM97:ic,GargA08:fossacs,GargP06:csfw,Goldblatt81:mlq,Goldblatt10:jlc,KrishnaswamiB11:lics,McbrideP08:jfp,Nakano00:lics,Nakano01:tacs}. They can be split roughly into two main groups. The first one---e.g., \cite{BirkedalMSS12:lmcs,KrishnaswamiB11:lics,Nakano00:lics,Nakano01:tacs}---concerns $\iRLob$ and has already been mentioned in Section \ref{sec:fixtheorem}. The second one---e.g., \cite{Aczel01:mscs,BentonBP98:jfp,Curry52:jsl,FairtloughM97:ic,GargA08:fossacs,GargP06:csfw,Goldblatt81:mlq,Goldblatt10:jlc}---concerns the system which is denoted here as $\iLax$ (a.k.a. $\iCLl$, see \cite{BentonBP98:jfp}).  \cite[Section 7.6]{Goldblatt03:jal} and \cite{Goldblatt10:jlc} are good if incomplete overviews of most  relevant papers on this system---the most important omissions being perhaps  \cite[Section 7]{PfenningD01:mscs} and  \cite{Abadi07:gdp,GargP06:csfw}. See also \cite{McbrideP08:jfp} 
for a discussion of type systems with $\iR$ modalities from programmer's point of view.
\end{myremark}

\cite{Esakia06:jancl} gives the following reasons for the importance of $\mHC$:
\begin{itemize}
\item Its connection to $\iKM$: 
$\imHC$ is $\iKM$ minus the L\"{o}b axiom $\rlax$. Note that $\kwax$ and $\rlax$ are  equivalent in the presence of $\rax$ 
\item The connection with \emph{intuitionistic temporal logic ``Always \& Before'' possessing rich expressive possibilities} 
\item The fact that $\imHC$ can be obtained as a fragment of $\qInt$ (or, in Esakia's notation, $\qHC$)---quantified intuitionistic propositional calculus. This is similar to the encoding of $\imHC$ in the internal language of a topos, see the last point
\item The topological connection with \emph{Cantor's scattered spaces, notions of the limit and isolated point}. 
 This will be discussed at length in this section 
\item Finally, as mentioned above, $\mHC$ turns out to be a natural
  fragment of the topos logic. This last point builds on all the preceding ones and will be  discussed in Sec. \ref{sec:scatopoi}
\end{itemize}

As we can see in Table \ref{tab:semcond}, the conditions on the accessibility relation $\mcc$ imposed by the axioms of $\imHC$---the combination of $\rax$ and $\nxtaxb$---is $\iccs \subseteq \mcc \subseteq \icc$. 
A natural question to ask is whether it is possible to enforce syntactically that  $\mcc$ is even more closely determined by $\icc$ as one of the two borderline cases, i.e., either $\mcc = \icc$ or $\mcc = \iccs$.

For $\mcc = \icc$, the answer is obviously positive. This is achieved precisely by the axioms of the logic $\iTriv$, strengthening $\nxtaxb$ 
to $\tax$. In fact, this is a semantic counterpart of an observation in \cite{Esakia06:jancl} that enriching any Heyting algebra with a trivial operator $\ibox^{\iTriv} x \deq x$  yields an  $\mHC$-algebra. Note here that whenever $\ibox$ is an $\imHC$- or even $\iR$-operator, its associated $\ribox$ is a $\iTriv$-operator.  

For $\mcc = \iccs$, the answer is obviously negative. 
Irreflexivity is a typical example of a condition which cannot be defined by any purely modal axiom, see \cite{BlackburndRV01:ml}. Here is perhaps the most natural proof.

\begin{expl}
Consider a frame for $\mHC$ defined as 
$(\omega, \icc, \iccs)$ where $\icc$ is the natural order $\leq$ on $\omega$. 
The dual algebra 
contains as a subalgebra the two-element Boolean algebra with $\ibox\bot=\bot$, which is  the dual of 
a single $\mcc$-reflexive point. Hence, no modal axiom can define the class of $\mcc$-irreflexive frames over the class of frames for $\mHC$.
\end{expl}



However, an ``irreflexive'' $\mHC$-operator is clearly definable on any Heyting algebra obtained as the dual of an intuitionistic Kripke frame $(W,\icc)$: for any $A \in \upset{W}$, we take $\irrbox A \deq \{ w \in W \mid \reluc{\{w\}}{\icc} - \{w\} \subseteq A \}$; again, see Table \ref{tab:semcond} for notation. It is straightforward to note that for any $w \in W$, $w \in \irrbox A$ iff for any $B \in \upset{W}$, $w \in (B \to A) \cup B$. This observation actually explains the shape of axiom $\nxtaxb$. Hence, 
given any complete Heyting algebra $\gH$, define its \emph{point-free  coderivative} \cite{EsakiaJP00:apal,Simmons82:lnm,Simmons:here} as $\irrbox \hml \deq \bigwedge\limits_{\iml \in \gH}(\iml \vee (\iml \to \hml))$.

\begin{prop} \label{prop:mhc}
For any complete Heyting algebra, its point-free coderivative is an $\mHC$-operator.
\end{prop}

\begin{proof}
A rather easy exercise for the reader; can be also extracted from the proof of \cite[Proposition 5]{Esakia06:jancl}.
\end{proof}

There is a slightly different description of $\irrbox$. Take a Heyting algebra $\gH$ and $\hml, \iml \in \gH$ s.t. $\iml \leq \hml$. $\hml$ is \emph{$\iml$-dense} or \emph{dense in $[\iml,\top]$} if for any $\jml \in \gH$, we get that $\hml \wedge \jml = \iml$ implies $\jml = \iml$. Note that the standard topological notion of density can be considered a special case: 
an open set is topologically dense iff it is $\bot$-dense in the Heyting algebra of open sets of the space. The following was observed, e.g., in \cite{EsakiaJP00:apal}:  

\begin{fact}
For any Heyting algebra $\gH$ and any $\hml \geq \iml \in \gH$, $\hml$ is $\iml$-dense iff there exists $\jml \in \gH$ s.t. $\hml = \jml \vee (\jml \to \iml)$.
\end{fact}


\begin{cor} \label{cor:simmonsdef}
For any complete Heyting algebra $\gH$ and any $\iml \in \gH$, $$\irrbox\iml = \bigwedge \{ \hml \in \gH \mid \hml \geq \iml \text{ and  $\hml$ is  $\iml$-dense} \}.$$
\end{cor}

Why \emph{coderivative}? The reader is referred to a detailed account
by Simmons \cite{Simmons:here}. Briefly, recall that in topology the \emph{Cantor-Bendixson derivative} of a set $A$ is the set of those $x$ whose every neighbourhood contains a point of $A$ \emph{other than $x$}; the dual operator (hence \emph{co-derivative}) consists of those $x$ which have an open neighbourhood \emph{entirely contained in $A \cup \{x\}$} \cite{Esakia06:jancl}. As it turns out, this indeed  coincides with  $\irrbox$ for practically all sensible topological spaces:

\begin{thm}[Simmons] \label{th:tzero}
For any $T_0$-space, its co-derivative operator coincides with the point-free coderivative $\irrbox$ on the Heyting algebra of open sets. 
\end{thm}

\begin{proof}
For $\irrbox$ defined as in Corollary \ref{cor:simmonsdef}, this was proved in \cite{Simmons82:lnm}.
\end{proof}

Obviously, as any intuitionistic Kripke frame with the Alexandroff topology given by $\upset{W}$ is $T_0$, we get that $\irrbox$ coincides with the dual of  topological derivative of this topology. It is, in fact, easier to prove directly than by Simmons' result. 

\begin{myremark}\label{}
One observation from \cite{Simmons82:lnm} is worth quoting here:  

\begin{quote}
for non-$T_0$-spaces the usual definition of isolated point does not quite capture the intended notion
\end{quote}
 
\noindent and hence for arbitrary spaces, the point-free definition of derivative given by $\irrbox$ seems in fact \emph{more adequate} than the standard one. The reader can verify this by extending the definition of intuitionistic Kripke frames to qosets rather than just posets 
and checking how both notions would fare in such a setting. 
\end{myremark}


\subsection{$\iKM$ and Scatteredness} \label{sec:scathey}

A complete Heyting algebra will be called \emph{scattered} if its coderivative $\irrbox$ is not only an $\mHC$-operator, but a $\KM$-operator. Recall Corollary \ref{cor:leofix} as an algebraic characterization of such a situation. 

\begin{prop} \label{prop:scat}
\mbox{\newline}
\begin{itemize}
\item For any topological space, its point-free coderivative 
is a $\KM$-operator \textbf{if} the space is scattered in the usual sense: that is, if each non-empty subset has an isolated point
\item For any $T_0$ topological space, its point-free coderivative is a $\KM$-operator \textbf{only if} the space is scattered
\end{itemize}
\end{prop}

\begin{proof}
A non-$T_0$-space can never be scattered, and for $T_0$-spaces, the point-free coderivative coincides with ordinary one as stated in Theorem \ref{th:tzero}. The remaining calculations are an exercise in point-set topology; in fact, the point-set part of this result has been shown first by Kuznetsov \cite{kuz79a,Muravitsky:here}. 
One can use an alternative characterization of scatteredness here: for any open set $A$ distinct from the whole space, $\irrbox A - A$ is non-empty.
\end{proof}


Let us summarize some of the results above:

\begin{cor}
A topological space  $\gT$ is scattered iff in the complete Heyting algebra of open sets of $\gT$, every polynomial in one $\irrbox$-guarded variable 
 has a fixed point.
\end{cor}
\begin{cor} \label{lm:kmkrcorr}
The following are equivalent for any $(W,\icc,\mcc)$: 
\begin{itemize}
\item $\mcc = \iccs$ and the Alexandroff topology $(W,\upset{W})$ is scattered
\item $(W,\icc,\mcc)$ is a frame for $\KM$
\item $\mcc = \iccs$ and it contains no infinite ascending chains
\end{itemize}
\end{cor}

\begin{proof}
We only need to prove the equivalence of the last two conditions. $\iccs \subseteq \mcc$ is, as observed, enforced by $\mHC$. Theorem \ref{th:ursini} gives the corresponding semantic condition for $\kwax$. Quite obviously, $\upset{W}$-Noetherianity forces irreflexivity of $\mcc$. Thus, whenever $(W,\icc,\mcc)$ is a frame for $\KM$  (that is, the join of $\mHC$ and $\iKW$), we have $\mcc = \iccs$. Moreover, rewriting the condition of $\upset{W}$-Noetherianity for $\mHC$-frames, 
we obtain that for any $A \in \upset{W}-\{W\}$, there is $w \in \reluc{A}{\iccs}-A$. 
Rewriting further, we obtain that  any $B \neq \emptyset$ s.t. $B = \reldc{B}{\icc}$ has a maximal element wrt $\icc$. But this means that \emph{any} nonempty subset of $W$ has a maximal $\icc$-element.
\end{proof}




\subsection{Completeness, Lattice Isomorphisms and Bimodal Translations} \label{sec:complibt}


Two important kinds of results have been missing from this overview so far. First, while I discussed Kripke \emph{correspondence} for modal logics (Table \ref{tab:semcond}, Theorem \ref{th:ursini} and Corollary \ref{lm:kmkrcorr}), I have not discussed \emph{completeness}. Second, I have not said much about lattices of extensions of L\"{o}b-like logics and their relatives---in particular, about generalizations of the Blok-Esakia Theorem. 

This section fixes both oversights. 
Rather than using original proofs of Kuznetsov, Muravitsky (for $\KM$) and Esakia (for $\mHC$), we are going to use corollaries of Wolter and Zakharyaschev's results on bimodal translations \cite{WolterZ97:al,WolterZ98:lw}, briefly discussed also in  \cite[Section 4]{WolterZ:here}. Their techniques allow to interpret intuitionistic modal logics as 
fragments of classical polymodal ones (cf. the discussion of \emph{implict vs. explicit epistemics} in \cite{Benthem91:la}). In the case of the Blok-Esakia theorem for $\mHC$, we will be able to see why axioms of $\mHC$ and $\cwGrz$ have to look the way they look in order to allow the classical counterpart to be unimodal rather than polymodal, as it happens in the more general framework of \cite{WolterZ97:al,WolterZ98:lw}. 

Take the bimodal language $\biF$ with operators $\wzi$ and $\wzm$. 
For any formula $A \in \mF$, I will write $A_\wzi$ (respectively $A_\wzm$) for $A$ with all occurrences of $\ibox$ replaced with $\wzi$ (respectively $\wzm$). $\wzm$ is the default counterpart of the original modality $\ibox$  and $\wzi$ encodes the intuitionistic poset order, hence the notation.\footnote{The reader has to be warned that the notation in this section differs somewhat from that in references like \cite{WolterZ97:al,WolterZ98:lw,WolterZ:here}.} The logic $\cSF$ is the normal logic axiomatized by the following axioms: $\clax$, $\kax_\wzi$, $\tax_\wzi$, $\kfax_\wzi$ 
and $\kax_\wzm$; in other words, it is what modal logicians would describe as the \emph{fusion} of $\cSFour_{\mathbf{i}}$ and $\cK_{\mathbf{m}}$. The logic $\cSFM$ is $\cSF$ extended with $$ \mixax: \qquad \wzm A \leftrightarrow \wzi\wzm\wzi A.$$ The logic $\cGFM$ is $\cSF$ extended with $\mixax$ and $\sgrzax_\wzi$. The translation $\wztr: \mF \to \biF$ prefixes every subformula in $(\cdot)_{\wzm}$ with $\wzi$. Of course, many occurrences of $\wzi$ in the translation $\wztr A$ can be removed relative to logics defined above: 

\begin{fact} \label{fact:wztreq}
The following equivalences belong to $\cSF$: $\wztr A \leftrightarrow \wzi \wztr A$, $\wztr(A \wedge B) \leftrightarrow (\wztr A \wedge \wztr B)$, $\wztr(A \vee B) \leftrightarrow (\wztr A \vee \wztr B)$; in $\cSFM$, we moreover have $\wztr(\ibox A) \leftrightarrow \wzm\wztr A$. 
\end{fact}

For any intuitionistic normal modal logic $\Gamma \supseteq \cSF$ and any bimodal normal logic $\Delta \subseteq \biF$, let 
\begin{itemize}
\item $\wzmin\Gamma \deq \cSFM \lpl \{\wztr A \mid A \in \Gamma\}$ 
\item $\wzmax\Gamma \deq \wzmin\Gamma \lpl \sgrzax_\wzi$ 
\item $\wzint\Delta \deq \{ A \in \mF \mid \wztr A \in \Delta\}$
\end{itemize}
$\Delta$ is \emph{a $\biF$-companion of $\Gamma$} if for any $A \in \mF$, $A \in \Gamma$ iff $\wztr A \in \Delta$, i.e., iff $\wzint\Delta = \Gamma$.

\begin{thm}[\cite{WolterZ97:al,WolterZ98:lw,WolterZ:here}] \label{th:wz}
Let $\Delta \supseteq \cSF$ be  a normal bimodal logic and $\Gamma\subseteq \mF$ be an intuitionistic normal logic.  Then
\begin{description}
\item[\Acls] $\wzint\Delta$ is an intuitionistic normal modal logic
\item[\Bcls] $\wzmin\Gamma$ and $\wzmax\Gamma$ are, respectively, the smallest and the greatest $\biF$-companions of $\Gamma$ containing $\mixax$
\item[\Ccls] $\wzint$ preserves decidability, Kripke completeness and the finite model property. If $\mixax \in \Delta$, $\wzint$ also preserves canonicity
\item[\Dcls] $\wzmin$ preserves canonicity
\item[\Ecls] $\wzmax$ preserves the finite model property 
\item[\Fcls] $\wzmax$ is an isomorphism from the lattice of normal extensions of $\iK$ onto the lattice of normal extensions of $\cGFM$
\item[\Gcls] $\Gamma \supseteq \iKFour$ has the finite model property whenever its $\biF$-companions over $\cSFM \lpl \kfax_{\wzm}$ include a canonical subframe logic  
\end{description}
\end{thm}

\begin{proof}
\Acls can be easily proved from Fact \ref{fact:wztreq}; note that we need the assumption we are above $\cSF$. \Bcls is a consequence of Theorem 27 in \cite{WolterZ97:al}. \Ccls, \Dcls and \Ecls are consequences of Proposition 29 and Theorem 30 in \cite{WolterZ97:al} and Theorems 11 and 12 in \cite{WolterZ98:lw}.\Fcls is a consequence of Corollary 28 in \cite{WolterZ97:al}. \Gcls is a consequence of Corollary 18 in \cite{WolterZ98:lw}.
\end{proof}

\cite{WolterZ97:al,WolterZ98:lw} illustrate on many examples how powerful these results are. 
As it turns out, they also have corollaries of immediate interest for us.

\begin{cor} \label{cor:mhcfmp}
$\mHC$ is canonical and has the finite model property.
\end{cor}

\begin{proof}
First, note that   $$\cSFM \lpl \wztr\rax = \cSFM \lpl  (\wzi A \to \wzm A)\supseteq \cSFM \lpl \kfax_{\wzm}.$$ Clearly, $\wzi A \to \wzm A$  is a Sahlqvist formula with an universal FO counterpart.  Furthermore,  $\cSFM \lpl \wztr\nxtaxb$ is the same logic as the extension of $\cSFM$ with 
\begin{equation} \label{eq:bifnext}
\wzm B \wedge \wzid C \to \wzi(\wzid C \vee \wzi B). 
\end{equation}
The latter is a simple Sahlqvist implication (cf. e.g., \cite[Definition 3.47]{BlackburndRV01:ml}). Applying the algorithm in the proof of Theorem 3.49 in \cite{BlackburndRV01:ml} and doing some FO-preprocessing, we get an universal formula
\begin{equation} \label{eq:fonext}
\forall y,z,w.(x \icc y \wedge x \icc z \to (z \icc y \vee (z \icc w \to x \mcc w)))
\end{equation}
(where $\icc$ is the accessibility relation corresponding to $\wzi$ and $\mcc$ is the accessibility relation corresponding to $\wzm$). Thus, $\wzmin\mHC$ is a canonical subframe logic over $\cSFM \lpl \kfax_{\wzm}$. Now, canonicity of $\mHC$ follows from \Ccls and the fmp from \Gcls of Theorem \ref{th:wz}.
\end{proof}

\begin{myremark}\label{rem:fonxtaxb}
It is worth noting that the semantic counterpart of $\nxtax$ from Table \ref{tab:semcond}, i.e., $\iccs \, \subseteq \, \mcc$ is equivalent to \refeq{eq:fonext} above. 
For one direction, substitute $x = y$ in \refeq{eq:fonext} and use poset properties. For the other direction, note that whenever $x \icc y$ and $x \icc z$ but $\neg(z\icc y)$, then $x \iccs z$, ergo $x \mcc z$. Now whenever $z \icc w$, we can use the interaction between $\mcc$ and $\icc$ as expressed by $\mixax$ (in fact, even a weaker axiom would do). 
\end{myremark}

Canonicity of $\mHC$ has been noted
, e.g., in \cite{Esakia06:jancl,Goldblatt81:mlq}. I was unable to locate references where the finite model property has been explicitly claimed. 
The following corollary can appear more surprising, as bimodal logics over $\biF$ do not even occur in its statement.

\begin{cor}\label{cor:blokesakia}
The lattice of normal extensions of $\mHC$ is isomorphic to the lattice of normal extensions of $\cwGrz$. The sublattice of normal extensions of $\KM$ is isomorphic to the lattice of normal extensions of $\cKW$.
\end{cor}

\begin{proof}
The heart of the proof is to notice that $\wzi A \leftrightarrow \wzm A \wedge A$ and $\grzax_{\wzm}$ are derivable in $\wzmax\mHC$; in fact, these two formulas axiomatize this logic over $\biF \lpl \kfax_{\wzm}$. Let us derive the first of them. For convenience, we will do it in the algebraic setting: 
\newcommand{\bywhat}[1]{ \qquad \qquad \text{by } {#1}}
\begin{align*}
\wzm A \wedge A \wedge \wzid\neg A \leq &  \bywhat{\tax_{\wzi}}\\ 
\wzm A \wedge \wzid(A \wedge \wzid \neg A) \leq & \bywhat{\text{\refeq{eq:bifnext}}} & \\ 
\wzi(\wzid(A \wedge  \wzid\neg A) \vee \wzi A) \leq & \bywhat{\tax_{\wzi}}  \\ 
\wzi(\wzid(A \wedge  \wzid\neg A) \vee A) = & \\
\wzi(\wzi(A \to \wzi A) \to A) \leq  & \, \wzi A \bywhat{\grzax_{\wzi}.} 
\end{align*}
  We get that $\wzmax\mHC$ is just a notational variant of $\wzmax\cwGrz$, with $\wzm$ being $\ibox$ and $\wzi$ being $\ribox$. This yields the first statement by clause \Fcls of Theorem \ref{th:wz}. For the second, it is enough to add the observation that over $\wzmax\mHC$, adding $\wztr\rlax$ is equivalent to adding $\kwax_{\wzm}$.
\end{proof}

The second statement of Corollary \ref{cor:blokesakia} above was first proved by Kuznetsov and Muravitsky in mid-1980's, see \cite{KuznetsovM86:sl,Muravitsky:here}. The first statement was announced in \cite{Esakia06:jancl} as follows:

\begin{quote}
Finally let us note that \dots the lattice Lat(mHC) of all extensions of mHC is isomorphic to the lattice Lat(K4.Grz) of all normal extensions of the modal system K4.Grz. However, a proof of this result requires additional considerations as the above algebraic machinery does not suffice for it.
\end{quote}

It seems that the proof has not been published so far.

\begin{cor}[\cite{Goldblatt81:mlq,KuznetsovM86:sl,mur81:mz}] \label{cor:kmfmp}
$\KM$ has the finite model property.
\end{cor}

\begin{proof}
The proof of Corollary \ref{cor:blokesakia} has established that $\wzmax\KM$ is just a notational variant of $\cKW$, with $\wzm$ being $\ibox$ and $\wzi$ being $\ribox$. Now use clause \Ccls of Theorem \ref{th:wz} and the finite model property for $\cKW$ (see, e.g., \cite{BlackburndRV01:ml,Boolos93:lop,ChagrovZ97:ml,Fine85:jsl,Moss07:jpl} for references).
\end{proof}

\begin{myremark}\label{}
Note that we could also prove Corollary \ref{cor:mhcfmp} in an analogous way to Corollary \ref{cor:kmfmp}, using the fmp of $\cwGrz$ established explicitly by Amerbauer \cite{Amerbauer96:sl}. The latter is actually a direct consequence of $\cwGrz$ being a transitive subframe logic \cite{Fine85:jsl,Litak07:bsl}. However, I believe that the proof of Corollary \ref{cor:mhcfmp} provided above has some additional value: we obtained a convenient form of $\wztr\nxtaxb$---which we actually used in the proof of Corollary \ref{cor:blokesakia}---together with its FO translation, which also provides some additional insight, as discussed in Remark \ref{rem:fonxtaxb}. 
\end{myremark}

\begin{myremark}\label{}
It could be an interesting exercise---and very much in the spirit of the Tbilisi school---to show that the above-discussed results of \cite{WolterZ97:al,WolterZ98:lw,WolterZ:here} survive when the base bimodal logic is weakened from $\cSF$  to the fusion of $\cK_{\mathbf{i}} \lpl (A \wedge \wzi A \to \wzi\wzi A)$ with $\cK_{\mathbf{m}}$ and the translation $\wztr$ is modified to $\wztr^*$ replacing every subformula $A$ with $\wztr^*A \wedge \wzi\wztr^*A$. On the other hand, it is not obvious how much generality would be really gained in this way. Note that using Wolter and Zakharyaschev's original results we were able to investigate lattices of logics which are \textbf{not} extensions of $\cSFour$, such as $\cwGrz$ in Corollary \ref{cor:blokesakia} above.
\end{myremark}





\section{Scattered Toposes} \label{sec:scatopoi}

\renewcommand{\cmp}{\circ}

We are ready to discuss the \emph{topos of trees} 
of  \cite{BirkedalMSS12:lmcs}, \emph{scattered toposes} of  \cite{EsakiaJP00:apal} and the relationship between fixpoint results 
in both papers.

\subsection{Preliminaries on Topos Logic}

Just like Section \ref{sec:complibt} assumed certain familiarity with technicalities of modal logic, this section in turn assumes some familiarity with basics of category theory---mostly with the notions of a ccc (cartesian closed category), a functor and a natural transformation. Those readers who know more than that, in particular understand well the internal logic of a topos,  can probably skip this subsection. Due to obvious space constraints, the presentation has to be rather abstract and example-free; see \cite{Goldblatt06:topoi,johnstone2002sketches,MacLaneM92:sigl} for more examples and motivation. Furthermore, like any presentation of topos theory by logicians and for logicians, it can be accused of neglecting spatial intuitions. See, e.g., \cite{McLarty90:bjps} for a passionate polemic with the view that toposes were invented to generalize set-theoretical foundations of mathematics.\footnote{Speaking of \cite{McLarty90:bjps}, footnote 4 provides an argument that the plural form intended by Grothendieck was \emph{toposes} rather than \emph{topoi}. I stick to the same convention, also because---as a quick Google search shows---the form \emph{toposes} is used mostly by mathematicians, whereas \emph{topoi} seems prevalent for unrelated notions in the humanities. Besides, this was the form used by Leo, Mamuka and Dito.} Nevertheless,  applications of toposes in fields like algebraic geometry or foundations of physics or their actual historical origins are not directly relevant here. My aim is a minimalist presentation focusing on the contrast between the logic of a topos and that of a ccc, but also making clear how the Beth-Kripke-Joyal semantics is related to more familiar ones for the intuitionistic predicate logic. Of all accounts in the literature, the one in \cite{LambekS86:ihocl} is probably closest to this goal. 

Let $\ctC$ be a ccc with  the terminal object $\trmo$ and for any $Y \in \ctC$, let $\finim_Y$ be the unique element of $\ctC[Y,\trmo]$. I use the obvious notation for (finite) products, coproducts (whenever they exist, but in a topos they always do), their associated morphisms and I denote the ccc evaluation mapping $B^A\times A \to B$ as $\eval_{A,B}$. Recall that $\ctC$ is an \emph{elementary topos}  if there exists  an  object $\Omega \in \ctC$  s.t. $(\Omega, \trmo \stackrel{\topT}{\to} \Omega)$ is a \emph{subobject classifier}, i.e., for any \emph{monic} (left-cancellable morphism) $Y \stackrel{f}{\monar} X$ there exists exactly one mapping $X \stackrel{\chrm{f}}{\to} \Omega$ s.t. 
we have a pullback diagram:
\[
\vcenter{
\xymatrix{
Y \ar[d]_{\finim_Y}\ar@{>->}[r]^{f}_<<{\pb} & X \ar[d]^{\chrm{f}}\\
\trmo \ar[r]^{\topT} & \Omega
}}
\]

\noindent  As observed by C. Juul Mikkelsen, this definition already implies that $\ctC$ is \emph{bicartesian closed}, where the latter notion is defined as in, e.g., \cite{LambekS86:ihocl}; 
see \cite[Section 4.3]{Goldblatt06:topoi} for references.

Before we proceed with  formal  definitions, some general discussion can be helpful. In every category, topos or not, (equivalence classes of) monics into $X$ are abstract counterparts of subsets of $X$; in fact, they are called \emph{subobjects of $X$}, just like morphisms $\trmo \to X$ are \emph{global elements} of $X$. 
Global elements (or equivalence classes thereof) can be considered as special cases of subobjects: think of the usual identification of an element $x \in X$ with the subset $\{x\}$.\footnote{Note that toposes very rarely happen to mimic sets in having enough global elements to determine all subobjects; such special toposes are called \emph{well-pointed}.}  We can  go further and define a \emph{generalized element} of $X$ as \emph{any} morphism $A \to X$, which is then called \emph{$A$-based} or \emph{defined over $A$}. See \cite[Section V]{MacLaneM92:sigl} for a lucid and brief discussion of those notions.

 In particular, the global elements of $\Omega$ can be identified with logical constants, $X$-based generalized elements of $\Omega$ with predicates over $X$ (i.e., formulas with a single free variable from $X$) and $n$-ary propositional connectives with morphisms  $\Omega^n \to \Omega$. 
Set $\botT \deq \chrm{\finim_0}$, $\negT \deq \chrm{\botT}$, $\andT \deq \chrm{\prdar{\topT,\topT}}$ and $\orT \deq \chrm{\cprar{\prdar{\topT_\Omega,\id_\Omega},\prdar{\id_\Omega,\topT_\Omega}}}$. Recall that for any $X\in\ctE$, $\topT_X$ stands for $\topT \cmp \finim_X$ and $\eqcT_X$ stands for  $\chrm{\prdar{\id_X,\id_X}}$. The latter allows to define \emph{internal equality for generalized elements of type $X$} as $\sigma \eqT \tau \deq \eqcT_X \cmp \prdar{\sigma,\tau}$. That is, if $A \stackrel{\sigma}{\to} X$ and $B \stackrel{\tau}{\to} X$ are generalized elements of $X$, then $\sigma \eqT  \tau$ is a generalized element of $\Omega$ defined over $A \times B$. For $\Omega$, we can define not only $\eqcT_\Omega$, but also $\leqcT_\Omega$ as the equalizer of 
$\vcenter{
\xymatrix{
\Omega \times \Omega \ar@<1ex>[r]^{\andT}\ar@<-1ex>[r]_{\prdf} & \Omega
}}$. 
Implication, the only remaining intuitionistic connective, can be now defined as $\impT \deq \chrm{\leqcT_\Omega}$. 

Thus, in toposes one can reduce reasoning about the poset  of subobjects of any given object $X \in \ctC$ (in fact, whenever $\ctC$ is a topos, this poset is always a lattice and even a Heyting algebra---see \cite[Theorem IV.8.1]{MacLaneM92:sigl}) to reasoning about  $\ctC[X,\Omega]$ and further still to reasoning about \emph{an internal Heyting algebra} provided by a suitable exponential object. What this means is: in any category,  monics into $X$ have a natural preorder defined as $f \subseteq g$ if $f$ factors through $g$, i.e., there is a morphism $h$ s.t. $f = g \cmp h$. Dividing by equivalence classes with respect to $\subseteq$, we get a category-theoretic generalization of the poset of subsets of $X$ ordered by inclusion. In general, without understanding the global structure of $\ctC$, we are not likely to learn much about these posets of subobjects. But in a topos, the poset of subobjects of $X$ is isomorphic to something more tangible: namely, to $\ctC[X,\Omega]$, i.e., $X$-based generalized elements of $\Omega$. Think of the usual identification of subsets of $X$ with elements of $2^X$. 
Here is also where first- and higher-order aspects of the internal logic come into play. 

If $\ctC$ is a category with products, a \emph{power object} of $X \in \ctC$ is a pair $\vcenter{
\xymatrix{(\pwoC{X}, \memoC_X \ar@{>->}[r]^{\memmC_X} & \pwoC{X} \times X)}}$ s.t. for any $\vcenter{
\xymatrix{(Y, R \ar@{>->}[r]^{r} & Y \times X)}}$ there exists exactly one $\vcenter{
\xymatrix{Y \ar@{->}[r]^{f_r} & \pwoC{X}}}$ for which there is a pullback
\[
\vcenter{
\xymatrix{
R \ar[d]\ar@{>->}[r]^{r}_<<{\pb} & Y \times X \ar[d]^{f_r \times \id_X}\\
\memoC_X\ar@{>->}[r]^{\memmC_X} & \pwoC{X} \times X
}}
\]

As shown, e.g., in \cite[Theorem 4.7.1]{Goldblatt06:topoi}, in any topos we can take $\pwoC{X}$ to be $\Omega^X$ 
and the subobject  $\vcenter{
\xymatrix{\memoC_X \ar@{>->}[r]^{\memmC_X} & \Omega^X  \times X}}$ can be obtained by pulling back $\topT$ along $\Omega^X \times X \stackrel{\eval_{X,\Omega}}{\to} \Omega$. Thus, we see that in a topos, the notions of power object, subobject classifier and exponential object are indeed well-matched and we can define the membership predicate $\sigma \inT \tau$ for a pair of generalized elements  $(A \stackrel{\sigma}{\to} X, B \stackrel{\tau}{\to} \pwoC{X})$  as $\eval_{X,\Omega} \cmp \prdar{\tau,\sigma}$. 
We are now ready for a single definition formalizing the whole discussion above and more (see \cite[Sec. VI.5-7]{MacLaneM92:sigl} and also \cite{Crole93:c4t,LambekS86:ihocl}):

\begin{defn}[The  Mitchell-B\`{e}nabou languague]
  Consider a topos $\ctC$.  The  \emph{collection of ground types} and the \emph{signature} of the  Mitchell-B\`{e}nabou language of $\ctC$ are defined, respectively, as 
\begin{align*}
\Grnd{\ctC} \deq & \{ \synt{E} \mid  E \in \ctC \} \\
\Sg{\ctC} \deq & \{ \synt{f}:\synt{F_1},\dots,\synt{F_n} \to \synt{E} \mid f \in \ctC[F_1\times\dots\times F_n,E] \}
\end{align*}
(instead of $\synt{k}:\synt{\trmo} \to \synt{E}$ I will write $\synt{k}:\synt{E}$) 
and  \emph{the full collection of types $\Typs{\ctC}$} is 
$A, B ::= \synt{E} \mid \chtrm \mid \chom \mid A \chtim B \mid \chexp{A}{B}$ 
where $\synt{E} \in \Grnd{\ctC}$. 
$\chexp{A}{\chom}$ can be also written as $\chpow{A}$.
Fix, moreover, a supply of term variables $x, y, z \dots \in \tVar$. \emph{The collection of terms $\Trms{\ctC}$ over $\Sg{\ctC}$} is defined as 
\begin{equation*}
M, N
 ::=  x \mid  \synt{f}\overline{M} \mid \chfin \mid M \eqT N \mid \chprd{M}{N} \mid  \chprdf{M} \mid  \chprds{M} \mid \lambda x:A.M \mid \chof{M}{N} 
\end{equation*}
where $x \in \tVar$ and $\synt{f} \in \Sg{\ctC}$ is of suitable arity. The \emph{typing rules} and some standard abbreviations (including all logical connectives) of the language are defined in Table \ref{tab:mbrules}.
Interpretation of types, contexts and terms-in-context in $\ctC$ is given in Table \ref{tab:mbint}. 
\end{defn}

\def\bgskip{2.8mm}
\def\tbskip{3.8mm}

\begin{table}
\begin{tabular}{>{$}c<{$}>{$}c<{$}>{$}c<{$}>{$}c<{$}}
 \hline
\inferrule{ }{\Gamma, x:A \chr x:A} &  \inferrule{\Gamma \chr M : A \quad \Gamma \chr N : A}{\Gamma \chr M \eqT N : \chom} & \inferrule{ }{\Gamma\chr\chfin : \chtrm} &  \inferrule{\Gamma \chr M: A \quad  \Gamma \chr N:B}{\Gamma\chr\chprd{M}{N}: A \chtim B} \\[\tbskip]
\inferrule{\Gamma,x:A \chr M:B}{\Gamma\chr\lambda x:A.M:\chexp{A}{B}} & \inferrule{\Gamma\chr M:\chexp{A}{B} \quad \Gamma\chr N:A}{\Gamma \chr \chof{M}{N}:B} & \inferrule{\Gamma\chr M: A \chtim B}{\Gamma\chr\chprdf{M}: A} & \inferrule{\Gamma\chr M: A \chtim B}{\Gamma\chr\chprds{M}: B} \\[\tbskip] 
\multicolumn{4}{c}{\inferrule{ \synt{f}:\synt{F_1},\dots,\synt{F_n} \to \synt{E} \in \Sg{\ctC} \quad \Gamma \chr M_1 : \synt{F_1} \dots \Gamma \chr M_n : \synt{F_n}}{\Gamma \chr \synt{f}M_1\dots M_n : \synt{E}}} \\[\bgskip] 
\\
\chtrue \deq \chfin \eqT \chfin & & \multicolumn{2}{l}{$\phi \wedge \psi \deq \chprd{\phi}{\psi} \eqT \chprd{\chtrue}{\chtrue}$}  \\
\multicolumn{2}{l}{$\forall x:A.\phi \deq \lambda x:A.\phi \eqT \lambda x:A.\chtrue$} & \phi \chto \psi \deq \phi \wedge \psi \eqT \phi & \\ 
\multicolumn{2}{l}{$\exists x:A. \phi \deq \forall t:\synt{\Omega}.((\forall x:A.\phi \chto t) \chto t)$} & \chfalse \deq \forall t:\synt{\Omega}.t & \\
\multicolumn{2}{l}{$\phi \vee \psi \deq \forall t:\synt{\Omega}.((\phi \chto t) \wedge (\psi \chto t) \chto t)$} & \neg\phi \deq \phi \chto \chfalse  & 
\\[\tbskip]\hline
\end{tabular}
\caption{\label{tab:mbrules}Typing rules and defined abbreviations of the  Mitchell-B\`{e}nabou language of $\ctC$}
\end{table}

\begin{table}
\begin{tabular}{>{$}c<{$}>{$}c<{$}>{$}c<{$}>{$}c<{$}>{$}c<{$}}\hline
\dent{\chE} \deq E & \dent{\chtrm} \deq \trmo & \dent{\chom} \deq \Omega & \dent{A \chtim B} \deq \dent{A} \times \dent{B} & \dent{\chexp{A}{B}} \deq \dent{B}^{\dent{A}} \\[\tbskip]
\multicolumn{5}{c}{$\dent{x_1:A_1,\dots,x_n:A_n} \deq \dent{A_1} \times \dots \times \dent{A_n}$} \\[\tbskip]
\multicolumn{3}{c}{$\inferrule{ }{\dent{\Gamma, x:A \chr x:A} = \pi: \dent{\Gamma} \times \dent{A} \to \dent{A}}$} &
\multicolumn{2}{c}{$
\inferrule{ }{\dent{\Gamma\chr\chfin : \chtrm} = \finim_{\dent{\Gamma}}}
$} 
\\[\tbskip]
\multicolumn{5}{c}{$\inferrule{ f \in \ctC[A_1 \times\dots \times A_n, B] \quad \dent{\Gamma \chr M_1 : A_1} =  \sigma_1: \dent{\Gamma} \to \dent{A_1}  \dots \dent{\Gamma \chr M_n : A_n} = \sigma_n: \dent{\Gamma} \to \dent{A_n} }{\dent{\Gamma \chr \synt{f}M_1\dots M_n : B} \deq f \cmp \langle \sigma_1,\dots,\sigma_n\rangle: \dent{\Gamma} \to B}$} \\[\tbskip]
 \multicolumn{5}{c}{$\inferrule{\dent{\Gamma \chr M : A} =  \sigma: \dent{\Gamma} \to \dent{A} \quad \dent{\Gamma \chr N : A} =  \tau: \dent{\Gamma} \to \dent{A}}{\dent{\Gamma \chr M \eqT N : \chom} \deq  \chrm{\prdar{\id_{\dent{\Gamma}},\id_{\dent{\Gamma}}}} \cmp \prdar{\sigma,\tau}}$} \\[\tbskip] 
\multicolumn{5}{c}{$\inferrule{\dent{\Gamma \chr M: A} = \sigma: \dent{\Gamma} \to \dent{A}  \quad  \dent{\Gamma \chr N:B} = \tau: \dent{\Gamma} \to \dent{B}}{\dent{\Gamma\chr\chprd{M}{N}: A \chtim B} \deq \langle\sigma, \tau\rangle: \dent{\Gamma} \to \dent{A} \times \dent{B}}$} \\[\tbskip] 
\multicolumn{3}{c}{$\inferrule{\dent{\Gamma \chr M: A \times B} = \sigma: \dent{\Gamma} \to \dent{A}  \times \dent{B}}{\dent{\Gamma\chr\chprdf{M}: A} \deq \prdf\cmp\sigma: \dent{\Gamma} \to \dent{A}}$} &
\multicolumn{2}{c}{$\inferrule{\dent{\Gamma \chr M: A \times B} = \sigma: \dent{\Gamma} \to \dent{A}  \times \dent{B}}{\dent{\Gamma\chr\chprds{M}: B} \deq \prds\cmp\sigma: \dent{\Gamma} \to \dent{B}}$} \\[\tbskip]
\multicolumn{5}{c}{$\inferrule{\dent{\Gamma,x:A \chr M:B} = \sigma:\dent{\Gamma}\times\dent{A} \to \dent{B}}{\dent{\Gamma\chr\lambda x:A.M:\chexp{A}{B}} \deq \currya{\dent{\Gamma} \times \dent{A},\dent{B}}{\sigma}:\dent{\Gamma} \to \dent{B}^{\dent{A}}}$} \\[\tbskip]
\multicolumn{5}{c}{$\inferrule{\dent{\Gamma\chr M:\chexp{A}{B}} = \sigma: \dent{\Gamma} \to \dent{B}^{\dent{A}} \quad \dent{\Gamma\chr N:A} = \tau: \dent{\Gamma} \to \dent{A}}{\dent{\Gamma \chr \chof{M}{N}:B} \deq \eval_{\dent{A},\dent{B}} \cmp \prdar{\tau,\sigma}}$} \\[\tbskip]\hline
\end{tabular}
\caption{\label{tab:mbint}Interpretation of types, contexts and terms-in-context}
\end{table}



\begin{defn}[Forcing for an elementary topos]
Assume $\Gamma = x_1:\synt{F_1}, \dots, x_n:\synt{F_n}$ and $\Gamma \chrmn \phi: \chom$. By $\subob{\langle x_1\dots x_n \rangle}{F_1 \times \dots \times F_n}{\phi}$, I will denote the pullback of the following diagram:
\[
\vcenter{
\xymatrix{
& \dent{\Gamma} \ar[d]^{\dent{\Gamma\chr\phi:\chom}}\\
\trmo \ar[r]^{\topT} & \Omega
}}
\]
Now for $F \stackrel{f_1}{\to} F_1, \dots, F \stackrel{f_n}{\to} F_n$ write $\tmod{F}{f_1,\dots,f_n}{\phi}$ if $F \stackrel{\langle f_1,\dots,f_n\rangle}{\to} \dent{\Gamma}$ factors through $\subob{\langle x_1\dots x_n \rangle}{F_1 \times \dots \times F_n}{\phi} \to \dent{\Gamma}$. In what follows, $f_1\cmp g,\dots,f_n\cmp g$ will be denoted by $\ofseq\cmp g$. Moreover, let $\dent{\Gamma\chr\phi:\chom} = \topT_{\dent{\Gamma}}$ be written as  $\conss{\Gamma}{\ctE}{\phi}$.

\end{defn}


\begin{fact} \label{prop:truthcond}
\mbox{\newline}
\begin{itemize}
\item $\tmod{F}{\ofseq}{\phi}$ iff $\dent{\Gamma\chr\phi:\chom} \cmp \ofseq = \topT_F$
\item
 $\conss{\Gamma}{\ctE}{\phi}$ iff 
for any $F \stackrel{\ofseq}{\to} \dent{\Gamma}$, it holds that $\tmod{F}{\ofseq}{\phi}$ 
\end{itemize}
\end{fact}

The following result, which can be found as Theorem VI.6.1 in \cite{MacLaneM92:sigl} or Theorem II.8.4 in \cite{LambekS86:ihocl}, connects the definition of forcing given above with more standard intuitionistic semantics:

\begin{thm}[Beth-Kripke-Joyal semantics in an elementary topos] \label{thm:bkjs}
Assume $F \stackrel{f_1}{\to} E_1, \dots, F \stackrel{f_n}{\to} E_n$ and $\Gamma = x_1:\chE_1 \dots x_n:\chE_n$.
\begin{itemize}
\item $\tmod{F}{\ofseq}{\phi \wedge \psi}$ iff $\tmod{F}{\ofseq}{\phi}$ and $\tmod{F}{\ofseq}{\psi}$
\item $\tmod{F}{\ofseq}{\phi \vee \psi}$ iff there are arrows $G \stackrel{g}{\to} F$ and $H \stackrel{h}{\to} F$ s.t. $G \cpp H \stackrel{\cpr{g}{h}}{\epim} F$ is epi, $\tmod{G}{\ofseq \cmp g}{\phi}$ and $\tmod{H}{\ofseq \cmp h}{\psi}$
\item $\tmod{F}{\ofseq}{\phi \chto \psi}$ iff for any $G \stackrel{g}{\to} F$ it holds that $\tmod{G}{\ofseq \cmp g}{\psi}$ whenever $\tmod{G}{\ofseq \cmp g}{\phi}$
\item $\tmod{F}{\ofseq}{\neg\phi}$ iff  for any $G \stackrel{g}{\to} F$, it holds that $G \isom \inio$ whenever $\tmod{G}{\ofseq \cmp g}{\phi}$ \newline
For the case of quantified formulas, note that $\Gamma \chr \forall x_{n+1}:\chE_{n+1}.\phi$ iff $\Gamma, x_{n+1}:\chE_{n+1} \chr \phi$. Same holds for $\Gamma \chr \exists x_{n+1}:\chE_{n+1}.\phi$. Then we have:
\item $\tmod{F}{\ofseq}{\forall x_{n+1}:\chE_{n+1}.\phi}$ iff for every $G \stackrel{g}{\to} F$ and every $G \stackrel{g'}{\to} E_{n+1}$ it holds that $\tmod{G}{f_1 \cmp g,\dots, f_n \cmp g, g'}{\phi}$
\item $\tmod{F}{\ofseq}{\exists x_{n+1}:\chE_{n+1}.\phi}$ iff there exist $G \stackrel{g'}{\to} E_{n+1}$ and an epi $G \stackrel{g}{\epim} F$ s.t. $\tmod{G}{f_1 \cmp g,\dots, f_n \cmp g,g'}{\phi}$
\item $\tmod{F}{\ofseq}{\sigma \eqT \tau}$ iff $\dent{\Gamma \chr \sigma:\chE} \cmp \ofseq = \dent{\Gamma \chr \tau:\chE} \cmp \ofseq$
\end{itemize}
\end{thm}

The clauses for $\exists$ and $\vee$ above resemble those of intuitionistic Beth semantics. This is why ``Beth-Kripke-Joyal'' seems a more appropriate name in the general case of an arbitrary elementary topos; see, e.g., \cite[Section 14.6]{Goldblatt06:topoi}. However, when the topos happens to be the topos of \emph{presheaves}, i.e., covariant functors into  $\Set$, on a given small category $\ctR$---in particular, a poset taken as a category---the definition of forcing can be significantly simplified.
\footnote{Reader should be warned that in most of categorical literature, presheaves are assumed to be \emph{contra}variant, but see, e.g., \cite{Ghilardi89:aml} for an example of the covariant convention.}

Perhaps the most straightforward account of this simplification can be found in 
 \cite{LambekS86:ihocl}. 
First,  the  clause for disjunction can be ``kripkefied'' for \emph{indecomposable} objects and the clause for existential quantifiers---for \emph{projective} ones \cite[Proposition 8.7]{LambekS86:ihocl}. Second, the second clause of Fact \ref{prop:truthcond} suggests that to check validity of a given judgement-in-context $\Gamma\chr\phi:\chom$ in a topos, it is enough to restrict attention to those $F$ which belong to \emph{a generating set} for a given topos. Third, by the Yoneda Lemma, in a topos of presheaves $\Set^\ctR$ for an arbitrary small category $\ctR$, objects of the form $\hm{R}{C} \deq \ctR[C,-]$ for any given $C \in \ctR$ satisfy all these conditions: they are projective, indecomposable and do form a generating set. Moreover, also by the Yoneda Lemma,  elements of $\Set^\ctR[\hm{R}{C},F]$  are in $1-1$ correspondence with elements of $F(C)$:
\begin{align*}
\Set^\ctR[\hm{R}{C},F] \ni f & \to \ylt{f} \deq f_C(\id_C) \\
F(C) \times \hm{R}{C} \ni (c,h) & \to \ylm{c}(h) \deq Fh(c) 
\end{align*}
Note also that in clauses like the one for $\chto$, we can restrict attention to those  $G \stackrel{g}{\to} F$ whose source $G$ lies in the generating set. In the case of $\Set^{\ctR}$, this means replacing  $G \stackrel{g}{\to} \hm{R}{C}$ with elements of $\Set^{\ctR}[\hm{R}{D},\hm{R}{C}]$. But, by the Yoneda Lemma again, these can be replaced with arrows in $\ctR[C,D]$ (note the change of direction!).

Taking all this into account, we can obtain the following modified version of the semantics---this time properly ''Kripkean'' (see \cite[Proposition 9.3]{LambekS86:ihocl}).


\begin{cor}[Kripke-Joyal semantics in a topos of presheaves] \label{cor:kjs}
  Let $\ctR$ be a small category, $F_1,\dots,F_n \in \Set^{\ctR}$, $C \in \ctR$, $c_1 \in F_1(C)$, \dots  $c_n \in F_n(C)$ , $\Gamma = c_1:\chF_1, \dots, c_n:\chF_n$ and $\Gamma \chr \phi:\chom$.  Write $\tmod{C}{\ocseq}{\phi}$ for $\tmod{\hm{R}{C}}{\ylm{c}_1,\dots,\ylm{c}_n}{\phi}$. Given any $f \in \ctR[C,D]$, write $f(\ocseq)$ for $\ylm{c}_1(f),\dots,\ylm{c}_n(f)$---that is, $Ff(c_1), \dots, Ff(c_n)$. Then we have:
  \begin{itemize}
\item $\tmod{C}{\ocseq}{\phi \wedge \psi}$ iff $\tmod{C}{\ocseq}{\phi}$ and $\tmod{C}{\ocseq}{\psi}$
\item $\tmod{C}{\ocseq}{\phi \vee \psi}$ iff $\tmod{C}{\ocseq}{\phi}$ or $\tmod{C}{\ocseq}{\psi}$
\item $\tmod{C}{\ocseq}{\phi \chto \psi}$ iff for any  $f \in \ctR[C,D]$, 
$\tmod{D}{f(\ocseq)}{\psi}$ whenever $\tmod{D}{f(\ocseq)}{\phi}$
\item $\tmod{C}{\ocseq}{\neg\phi}$ iff  for any  $f \in \ctR[C,D]$, it does not  hold that $\tmod{D}{f(\ocseq)}{\phi}$ 
\item $\tmod{C}{\ocseq}{\forall x_{n+1}:\chF_{n+1}.\phi}$ iff for every $f \in \ctR[C,D]$ and $d \in F_{n+1}(D)$, it holds that $\tmod{D}{f(\ocseq),d}{\phi}$
\item $\tmod{C}{\ocseq}{\exists x_{n+1}:\chF_{n+1}.\phi}$ iff there exist $d \in F_{n+1}(C)$ s.t. $\tmod{C}{\ocseq,d}{\phi}$
  \end{itemize}
\end{cor}

\cite{Ghilardi89:aml}  uses toposes of presheaves as a generalization of  Kripke semantics for the intuitionistic first-order logic to prove incompleteness results. 
 Of numerous follow-ups of that work, let me just mention \cite{NagaokaI97:mlq,Skvortsov12:sl}.  Let us also note 
 that the derivation of Corollary \ref{cor:kjs} from Theorem \ref{thm:bkjs} takes a somewhat more roundabout route in \cite{MacLaneM92:sigl}: 
 toposes of presheaves are handled there as a subclass of toposes of \emph{sheaves on a site}.

\subsection{Non-expansive Morphisms, Fixpoints and  Scattered Toposes} \label{sec:scafix}

Let $\ctE$ be an elementary topos. 
Call an endomorphism $f \in \ctE[X,X]$ \emph{unchanging}  \cite{EsakiaJP00:apal} or \emph{non-expansive}  if 
\[
\conss{}{\ctE}{\forall x,y:\synt{X}.(\synt{f}x \eqT \synt{f}y \chto x \eqT y) \chto x \eqT y}.
\]
As noted in \cite{EsakiaJP00:apal}, in a boolean setting \emph{non-expansive} means just \emph{constant}: negate the sentence and play with boolean laws. Obviously then a \emph{classical} proof that a non-expansive endomorphism on a non-empty set has a unique fixed point does not carry much computational content. In a constructive setting, however, the situation is  different. 




Assume $\chr \phi, \psi: \chom^{\synt{X}}$ and $f \in \ctE[X,X]$ and define:

\begin{align*}
\Isubt(\phi) \deq & 
\forall x,y:\synt{X}(\chof{\phi}{x} \chand \chof{\phi}{y} \chto x \eqT y)\\
\Icont{\phi}{\psi} \deq &  \forall x:\synt{X}.(\chof{\phi}{x}\chto\chof{\psi}{x}) \\
\Imaxst(\phi) \deq & \Isubt(\phi) \chand \forall \alpha:\chom^{\synt{X}}.(\Isubt(\alpha) \chand \Icont{\phi}{\alpha} \chto \Icont{\alpha}{\phi}) \\
\Icontr(\synt{f}) \deq & \forall x,y:\synt{X}.(\synt{f}x \eqT \synt{f}y \chto x \eqT y) \chto x \eqT y \\
\Ifix_{\synt{f}} \deq & \lambda x:\synt{X}.(x \eqT \synt{f}x)
\end{align*}


With this apparatus, we can state the main Theorem of Section 3 of \cite{EsakiaJP00:apal}:

\begin{theorem} \label{th:scanonexp}
Assume $\ctE$ is an elementary topos and $f \in \ctE[X,X]$
. Then
\[
\conss{}{\ctE} \Icontr(\synt{f}) \chto \Imaxst(\Ifix_{\synt{f}}).
\]
\end{theorem}

\begin{myremark}\label{}
A proof formalized in the Coq proof assistant is available from the author, see Remark \ref{rem:coq}. Those who would like to try a manual yet rigorous proof in the  Mitchell-B\`{e}nabou language should do  first Exercise 5 in \cite[p. 139]{LambekS86:ihocl} 
and then formalize the proof  in \cite[p. 105]{EsakiaJP00:apal} using all the abbreviations given above. 
\end{myremark}

\begin{myremark}
 Some of derivations in
  \cites{Esakia98:bsl,EsakiaJP00:apal}, especially when fully
  formalized in a proof assistant, make explicit an interesting fact:
  certain axioms governing  quantification
  which are not generally intutionistically valid (e.g.,
  the Kuroda axiom and its  generalization due to Casari), become valid in
  presence of a $\iKM$ modality, even though these axioms do not
  involve any modalities at all in their statement. Some related
  observations are made by \cite{Biraben12:bsl} (a reference I became aware of having written the
  published version of this overview).
\end{myremark}

In words, this result says: \emph{the fixpoints of a non-expansive endomorphism form a maximal subterminal subobject}.\footnote{The corresponding theorem in \cite{EsakiaJP00:apal} contained also an  additional statement about density of the support of the fixed-point subobject, but this does not seem essential for us here.} The syntactic shape of $\Isubt{f}$ easily suggests that subterminality is the internal counterpart of ``being of cardinality at most one'', i.e., uniqueness of fixpoints. 
However, the situation with existence is more complicated. 
 First of all, toposes of presheaves can differ significantly from the topos of sets in having non-trivial objects with \emph{no global elements whatsoever}. More importantly, even being inhabited is not enough to ensure maximal subterminal objects are global elements.
\begin{expl}[\cite{EsakiaJP00:apal}] \label{ex:nofix}
Consider the topos of presheaves on  $(\omega+1,\unrhd)$, where $\unrhd$ is the converse of the standard ordinal order. Presheaf $X$ defined as $X(n) = n +1$ and $X(\omega) = \omega$ with $X(\beta \unrhd \alpha)(n) = min(n,\alpha)$ is clearly inhabited. Furthermore, $f: X \to X$ defined as $f_n(i) = min(i+1,n)$ and $f_\omega(i)=i+1$ is a non-expansive endomorphism. Yet it fails to have a fixpoint---i.e., a global element $\trmo \stackrel{c}{\to} X$ s.t. $f \cmp c = c$. 
\end{expl}
Of course, we can do better in special cases.

\begin{cor} \label{cor:injective}
Whenever 
$X \in \ctE$ is s.t. any maximal subterminal subobject of $X$ is a global element (for example, $X$ is an injective object), 
there exists $\finim \stackrel{c}{\to} X$ s.t. $f \cmp c = c$  for any non-expansive $f \in \ctE[X,X]$.
\end{cor}

We could try to express unique existence in the internal logic using the standard abbreviation $\exists!$ for ``exists exactly one''. However, as kindly pointed out by Thomas Streicher, this abbreviation works as intended in toposes of presheaves, but not necessarily in arbitrary ones. 




But where is the place for  a modality in all this? Say that $\anymo: \Omega \to \Omega$ is a \emph{strong L\"{o}b operator} if 
$
\conss{}{\ctE}{\forall p:\synt{\Omega}.(\chanymo p \chto p) \chto p}$.
 Also, call a morphism $f \in \ctE[X,Z]$ \emph{$\anymo$-contractive} if \, 
$\conss{}{\ctE}{\forall x,y:\synt{X}.\chanymo(x \eqT y) \chto (\synt{f}x \eqT \synt{f}y)}$:

\begin{cor}
Let $\anymo: \Omega \to \Omega$ be a strong L\"{o}b operator, $f \in \ctE[X,X]$, 
and assume that $f$ is $\anymo$-contractive. Then  $f$ is non-expansive and hence its subobject of fixed points is a maximal subterminal one. 
\end{cor}

\begin{proof}
  We have that 
$
\conss{}{\ctE}{\forall p,q:\synt{\Omega}.(\chanymo p\chto(p\vee q)) \chto ((q \chto p)\chto p)}.
$
 In fact, this is an equivalent form of $\rlax$---cf. the proof of Theorem 2(iv) in \cite{EsakiaJP00:apal}. 
Now substitute $x \eqT y$ for  $p$ and $\synt{f}x \eqT \synt{f}y$ for $q$ to get the result.
\end{proof}

\cite{EsakiaJP00:apal} states the result only for a special case of contractiveness and a special subclass of toposes (introduced below) but this generalization  is straightforward. As before, we can derive the conclusion about the existence of unique fixed points \emph{whenever every maximal subterminal object  of $X$ happens to be a global element}---e.g., whenever $X$ is injective.

Define $\chobmo \phi \deq \forall t:\chom.(t \vee (t \chto \phi))$, i.e., an internalized 
 coderivative. 
We have the following counterpart of Proposition \ref{prop:mhc}:

\begin{prop}
In any elementary topos $\ctE$,
  we have $\conss{}{\ctE}{\forall p:\chom.p \chto \chobmo p}$ and  $\conss{}{\ctE}{\forall p,q:\chom.\chobmo p \chto (q \vee (q \chto p))}$.
\end{prop}

A \emph{scattered topos} is defined analogously to scattered locales or  Heyting algebras in Section \ref{sec:scathey} by the validity of the only remaining $\iKM$ law, i.e., the  axiom 
$
\forall p:\chom.(\chobmo p \chto p) \chto p.
$
Thus, scattered toposes are those where $\chobmo$ is a strong L\"{o}b operator. This notion turns out to have several equivalent characterizations, 
 see \cite{EsakiaJP00:apal}. Let us discuss in detail here another one for the special case of $\Set^\ctR$:

\begin{thm} \label{th:scatpre}
Let $\ctR = (W,\icc)$ be a poset. Then $\Set^\ctR$ is scattered iff $ (W,\icc,\iccs)$ satisfies any of the equivalent conditions in Corollary \ref{lm:kmkrcorr}.
\end{thm}

\begin{proof}

For a direct proof
, it is useful to compute the semantic meaning of $\chobmo$. Define $\obmo: \Omega \to \Omega$ as $\dent{p:\chom \chr \chobmo p:\chom}$. Recall also that in a topos of the form $\Set^{\ctR}$ for $\ctR = (W,\icc)$, $\Omega(w)$ is equal to $\{ A \cap \reluc{\{w\}}{\icc} \mid A \in \upset{W} \}$. 
A morphism $f:\Omega \to \Omega$ is a natural transformation: a family of mappings $\{ f_w: \Omega(w) \to \Omega(w) \mid w \in W\}$ satisfying 
\begin{center}
$f_z(A \cap \reluc{\{z\}}{\icc}) = f_w(A) \cap \reluc{\{z\}}{\icc}$ for any $A \in \Omega(w)$, $z \cci w$.
\end{center}

Now let us note the following

\begin{fact}\label{prop:obmo}
For any topos of the form $\Set^\ctR$ where $\ctR = (W,\icc)$, for any $w \in W$ and for any $A \in \Omega(w)$ (i.e., $A$ an upward closed subset of $ \reluc{\{w\}}{\icc}$),
\[
\obmo_w(A) = \{z \cci w \; \mid  \; \reluc{\{z\}}{\iccs} \subseteq A\}.
\]
\end{fact}

The reader may want to consult Section \ref{sec:deriv} and Table \ref{tab:semcond} for the notation used above; in particular, recall that  $\reluc{\{z\}}{\iccs} =  \reluc{\{z\}}{\icc} - \{z\}$. Note also that we can add an atomic clause to Corollary \ref{cor:kjs} in the preceding section: 
\begin{fact}\label{fact:atclause}
$\tmod{w}{A}{t}$ (where $\Gamma = t:\chom$ and $A\in \Omega(w)$) iff $A = \reluc{\{w\}}{\icc}$.
\end{fact}

This fact, while rather basic, is worth an explicit proof, as it helps to put together several definitions and propositions above:

\begin{proof}[Proof of Fact \ref{fact:atclause}]
$\tmod{w}{A}{t}$ is an abbreviation for $\tmod{\hm{R}{w}}{\ylm{A}}{t}$, while this in turn can be reformulated as $\ylm{A} = \topT \cmp \finim_{\hm{R}{w}}$ (Fact \ref{prop:truthcond}). In particular, $\ylm{A}(w \icc w) =  \reluc{\{w\}}{\icc}$. But  $\ylm{A}(w \icc w) = A$.
\end{proof}

Fact \ref{fact:atclause} can be generalized with variable $t$ on the right hand side of the turnstile replaced with  arbitrary $\phi(t)$. Somewhat informally speaking,  $\tmod{w}{A}{\phi(t)}$ iff the value of $\phi(A)$ contains $\reluc{\{w\}}{\icc}$  (think of $\phi$ here as a polynomial on the Heyting algebra of upward closed subsets of $\reluc{\{w\}}{\icc}$).

Putting all this together, we get that $\tmod{w}{A}{(\chobmo p \chto p) \chto p}$ iff for any $z \cci w$,  it holds that $\tmod{z}{A \cap \reluc{\{z\}}{\icc}}{(\chobmo p \chto p)}$ implies $\reluc{\{z\}}{\icc} \subseteq A$. That is, $\{ z' \cci z \; \mid  \; \reluc{\{z'\}}{\iccs} \subseteq A \} \subseteq A$ only if $\reluc{\{z\}}{\icc} \subseteq A$ and in order for $\Set^{(W,\icc)}$ to be scattered this has to hold for any $w \in W$, any $z \cci w$ and any $A \in \Omega(w)$. But then the reasoning can be completed just like in the case of Corollary \ref{lm:kmkrcorr}. This finishes the proof of Theorem \ref{th:scatpre}.

\end{proof}

\subsection{Topos of Trees and Its Generalizations} \label{sec:toptrees}

\cite{BirkedalMSS12:lmcs} introduced the  \emph{topos of trees} (or
\emph{forests}) $\toptrees$, i.e., the topos of presheaves on
$(\omega,\unrhd)$, where $\unrhd$ is the converse of the usual order on
$\omega$. Let us begin with the observation that an axiomatization of the
$\imHC$-logic of the underlying frame of $\toptrees$ can be obtained
by means of techniques introduced in Section \ref{sec:complibt}.

\begin{theorem} \label{th:kmlc}
$\iKM \lpl \iLC$ is the propositional logic of $(\omega,\unrhd,>)$.  
\end{theorem}

\begin{proof}[Proof (sketch)]
 Just like in the first sentence of the proof of Corollary
 \ref{cor:kmfmp}, we note that $\wzmax{(\iKM \lpl \iLC)}$ is a just a
 notational variant of $\cKW\dtax\deq \cKW \lpl \dtax$. This logic in turn can be shown (using standard modal
 techniques) to be the logic of natural numbers with reverse strict
 order taken as a classical Kripke frame. The result now follows from properties
 of $\wzmax{\cdot}$ which are not quoted explicitly in Theorem
 \ref{th:wz}, but can be extracted immediately from the original references, for example \cite[Prop. 21]{WolterZ97:al}.
\end{proof}
 
Reasoning analogous to those in proofs of  Corollary \ref{cor:kjs} and Theorem
\ref{th:scatpre}  yields then that $\iKM \lpl \iLC$ coincides with the set of $\toptrees$-validities in the Mitchell-B\'{e}nabou
language restricted to  (topos-theoretical
counterparts of) connectives in $\iF$ together  with $\chobmo$, the
latter interpreted in in $\Set^{(\omega,\unrhd)}$ by $\obmo$ as
specified by Fact \ref{prop:obmo} above. Similar observations underly
recent proof-theoretical investigations of this logic in \cite{CloustonG14:sequent}.

Theorem \ref{th:scatpre} shows that the topos of trees is  scattered. 
$\Omega$-endomorphism ``$\rhd$'' (this notation here would risk
clashing with the one for a strict partial order and its converse)
defined  in \cite{BirkedalMSS12:lmcs} is easily seen to coincide with
$\obmo$. The Internal Banach Fixpoint Theorem 2.9 of \cite{BirkedalMSS12:lmcs} shows that $\obmo$-contractive mappings on arbitrary inhabited objects  in $\toptrees$ do have (unique) fixpoints.

Now, Example \ref{ex:nofix} above shows that such a strong statement is not valid in arbitrary scattered toposes of presheaves, even quite similar to $\toptrees$. The crucial Lemma 2.10 in  \cite{BirkedalMSS12:lmcs} is not amenable to far-reaching generalizations. 

However, \cite[Section 8]{BirkedalMSS12:lmcs} discusses a whole class of toposes together with a notion of a contractiveness guaranteeing fixpoint's existence. The class in question are \emph{sheaves} on complete  Heyting algebras with a well-founded basis  \cite{DiGianantonioM04:fossacs} rather than just presheaves on Noetherian partial orders---crucially, $\toptrees$ can be also seen  as such a sheaf topos---and the required notion of contractiveness is stronger than the one expressible in the internal logic. 

Let us elaborate on the last point. 
As we saw, toposes allow an internal interpretation of modalities as morphisms $\Omega \to \Omega$.  
Actually, from the ``propositions as predicates'' perspective,  any operation on subobjects  of a given object  is a ``local'' candidate for a modality. 
 However,  constructive or categorical logic is mostly about ``propositions as types''; see, e.g., \cite{AlechinaMdPR01:csl,BellinPR01:m4m,dePaivaR11:carnielli,BiermanP00:sl,PfenningD01:mscs,dePaivaR11:carnielli} for modal aspects. 
 This perspective 
works even with mild assumptions about the underlying category. 
In particular, \emph{algebraic type theories} require only finite products, whereas ccc's correspond to \emph{functional type theories} \cite{Crole93:c4t}: those whose type system is in fact that of Brouwerian semilattices of Remark \ref{rem:brouwerian} above. To see 
the details of this correspondence, just remove the rules for $\chom$, $\eqT$ and all abbreviations using these from Tables \ref{tab:mbrules} and \ref{tab:mbint}, then interpret conjunctions as products and implication as exponentation.

From this perspective, modalities correspond to endofunctors. 
 In particular, ND systems for $\iLax$ and $\iSFour$ are interpreted by, respectively, \emph{monads} and \emph{comonads}---see, e.g., references in Remark \ref{rem:monads}---and $\iSL$ yields a special subclass of \emph{pointed} or \emph{applicative}  functors. \cite{McbrideP08:jfp} 
 More precisely, one obtains a variant of  
  \cite[Definition 6.1]{BirkedalMSS12:lmcs}. Possible differences are: the modal assumption of normality forces only being monoidal wrt cartesian structure (cf. \cite{BellinPR01:m4m,dePaivaR11:carnielli}) rather than preservation of all finite limits as in the second clause of that definition; 
furthermore, the assumption of uniqueness in the first clause would
rely on exact reduction and conversion rules of the proof system. A
systematic study of such endofunctors in a cartesian setting---i.e.,
assuming only the presence conjunction among propositional connectives---has been
undertaken by \cite{MiliusL13:fics} under the name of \emph{guarded
  fixpoint categories}.  Other possible names for such endofunctors include \emph{contraction}, \emph{delay}, \emph{(strong) L\"{o}b}, $\iSL$ and \emph{MGRT}, the last being an abbreviation of the original name in \cite{BirkedalMSS12:lmcs}.

One can relate these two views on modalities. Whenever $F: \ctC \to \ctC$ is   monic-preserving and $\ctC$ has pullbacks, associate with a $F$-coalgebra $C \stackrel{\gamma}{\to} FC$ a modality $\pbcl{\gamma}{F}$ on subobjects  $M \stackrel{m}{\monar} C$:
 

\[
\vcenter{
\xymatrix{
\pbcl{\gamma}{F} M \ar[d]\ar@{>->}[r]^-{\pbcl{\gamma}{F} m}_<<{\!\!\!\!\!\!\!\pb} & C \ar[d]^{\gamma}\\
FM \ar@{>->}[r]^{Fm    } & FC
}}
\]

\noindent (see \cite{AdamekMMS12:fossacs} for the history of this diagram in papers on well-founded coalgebras). Furthermore, whenever $F$ is pointed (applicative), i.e., a $\iR$-endofunctor 
with $\point{F}: 1 \to F$ being the \emph{point} or \emph{unit} of $F$, $\point{F}_M$ is a subcoalgebra of $\point{F}_C$ for any $M \stackrel{m}{\monar} C$ 
and hence $m \leq \pbcl{\point{F}_C}{F} m$, i.e., the ``local'' translation of $\rax$ is universally valid. 
\begin{myremark}
\cite[Theorem 6.8]{BirkedalMSS12:lmcs} allows to isolate sufficient conditions
ensuring that the operator $\pbcl{\point{F}_C}{F}$ induced by
a $\iSL$-endofunctor $F: \ctC \to \ctC$ is a $\iSL$-modality: pullback-preservation
of $F$ 
 and $\ctC$  being a
topos. While it is not mentioned in \cite{BirkedalMSS12:lmcs}, one can
find natural counterexamples when such conditions are dropped!
\end{myremark}

The operation on subobjects of $C \in \toptrees$ induced by $\obmo: \Omega \to \Omega$ from Fact \ref{prop:obmo} is defined in an alternative way in \cite{BirkedalMSS12:lmcs}: as  $[\delbirk.\point{\delbirk}_C]$ for a delay endofunctor $\delbirk: \toptrees \to \toptrees$, whose action on objects is $(\delbirk C)(0) \deq \trmo$ and $(\delbirk C)(n+1) \deq (\delbirk C)(n)$. In a sense, $\delbirk$ can be called the \emph{Cantor-Bendixson endofunctor}. Factoring through it is the desired ``external'' notion of contractivity ensuring fixpoint's existence. Both notions nicely complement each other:
\begin{quote}
\dots the external notion provides for a simple algebraic theory of fixed points
for not only morphisms but also functors (see Section 2.6), whereas the internal notion is useful when working in the internal logic. \cite{BirkedalMSS12:lmcs} 
\end{quote}

Could \cite{EsakiaJP00:apal} have had more impact if the authors had a) employed the external perspective on modalities in addition to the internal one and b) had the hindsight of \cite{DiGianantonioM04:fossacs}? 
 This is  rather too counterfactual a question to  consider. Note also that what matters from the point of view of \cite{BirkedalMSS12:lmcs}---and Theoretical CS in general---is the use made of these external and internal L\"{o}b modalities. \cite[Section 3]{BirkedalMSS12:lmcs} constructs a model of a programming language with higher-order store and recursive types entirely inside the internal logic of $\toptrees$.  
 \cite[Section 4]{BirkedalMSS12:lmcs} provides semantic foundation
for dependent type theories extended  with a $\iSL$ modality  and guarded recursive types; this can be regarded as an extension of fixpoint results along the lines of Section \ref{sec:fixtheorem} above to predicate and higher-order constructive logics. \cite[Section 5]{BirkedalMSS12:lmcs} shows that a class of (ultra-)metric spaces commonly used in  modelling corecursion on streams is equivalent to a subcategory of $\toptrees$. 

Clearly, this is a large area rather overlooked by researchers on (intuitionistic) modal logic side. 
There is no space here to discuss my own  work in progress, e.g., on the Curry-Howard interpretation of $\mHC$, 
but let us conclude with a question from participants of ToLo III: 
is there a natural subclass of  internal  modalities in toposes (endomorphisms $\Omega \to \Omega$) inducing external modalities (endofunctors) in a generic way? 

\begin{myremark}
Earlier incarnations of this paper included also a question by Lars Birkedal: 
what are additional logical principles which would allow a scattered
topos to model not only guarded (co-)recursion, but also, e.g.,
countable nondeterminism? However, a recent work
\cite{BizjakBM14:rtlc} addressed the issue: expressing the additional property
needed for stating the adequacy of the logical relation requires more
than logical connectives introduced so far. It can be only stated in
presence of an additional modality right adjoint to double negation. 
\end{myremark}

\paragraph{Acknowledgements}

Special thanks are due to: the anonymous referee for careful reading and many helpful comments;  Guram Bezhanishvili, for an invitation to write this chapter, but also his tolerance, understanding and patience during the whole process; Alexander Kurz for his support, interest and numerous discussions; Mamuka Jibladze and the late Dito Pataraia, particularly for all the discussions we had in Kutaisi. 
Apart from this, the list of people who 
provided feedback and/or suggestions  include (in no particular order)
Lars Birkedal, Rasmus M\o gelberg,  Stefan Milius, Thomas Streicher,
Sam Staton, Alex Simpson, Gordon Plotkin, Johan van Benthem, Ramon
Jansana, Achim Jung, Andrew Pitts, Conor McBride, Nick Benton,
Guilhem Jaber, Dirk Pattinson,  Drew Moshier, Nick Bezhanishvili,
Alexei Muravitsky, Fredrik Nordvall Forsberg,  Ben Kavanagh, Oliver
Fasching, participants of Domains X, TYPES 2011, MGS 2012 and seminars
at the University of Leicester. For the extended version, I would also
like to acknowledge  Ranald Clouston, Rajeev Gor\'{e} and Ale\v{s} Bizjak.  
During early stages  of the write-up process, I was supported by the EPSRC grant EP/G041296/1. 

\printbibliography 










\end{document}

%% file: esakiaarxiv_macros.tex
\usepackage{amsthm}
\usepackage{amsfonts}
\usepackage{amsmath}
\usepackage{amssymb}
\usepackage{hyperref}
\usepackage{array}
\usepackage{multirow}
\usepackage{color}
\usepackage{mathpartir}
\usepackage{mathrsfs}
\usepackage{lscape}
\usepackage{
pgfrcs}
\usepackage{
pgf}
\usepackage{
tikz}
\usetikzlibrary{%
  shapes.misc,
  shapes.geometric,
  positioning,
  shadows%
}
\usepackage[all]{xy}
\xyoption{2cell}
\xyoption{curve}
\UseTwocells
\SelectTips{cm}{}
\usepackage{stmaryrd}

\usepackage{wasysym} 
\usepackage{pigpen}  

\newcommand{\pb}{\text{\large\pigpenfont J}}

\usepackage[framemethod=TikZ]{mdframed}

\input{envs_esakiafull}

\def\refeq#1{(\ref{#1})}

\newcommand{\tVar}{\mathsf{tVar}}
\newcommand{\Sg}[1]{\mathsf{Sg}(#1)}
\newcommand{\Grnd}[1]{\mathsf{Ground}(#1)}
\newcommand{\Typs}[1]{\mathsf{Types}(#1)}
\newcommand{\Trms}[1]{\mathsf{Terms}(#1)}

\newcommand{\deq}{:=}					

\newcommand{\ibox}{\Box}

\newcommand{\axst}[1]{\ensuremath{\mathsf{#1}}}
\newcommand{\intl}{\mathsf{int}}
\newcommand{\fullsig}{\ibox\intl}

\newcommand{\necr}{\axst{NEC}}
\newcommand{\mpr}{\axst{MP}}

\newcommand{\kax}{\axst{(nrm)}}
\newcommand{\kaxb}{\axst{(opr)}}
\newcommand{\kfax}{\axst{(trns)}}
\newcommand{\tax}{\axst{(refl)}}
\newcommand{\rax}{\axst{(r)}}
\newcommand{\raxb}{\axst{(fmap)}}
\newcommand{\cfax}{\axst{(bind)}}
\newcommand{\laxax}{\axst{(pll)}}
\newcommand{\kwax}{\axst{(wl\ddot{o}b)}}
\newcommand{\rlax}{\axst{(sl\ddot{o}b)}}
\newcommand{\gtax}{\axst{(glb)}}
\newcommand{\henkax}{\axst{(henk)}}
\newcommand{\grzax}{\axst{(grz)}}
\newcommand{\sgrzax}{\axst{(sgrz)}}
\newcommand{\nxtax}{\axst{(next)}}
\newcommand{\nxtaxb}{\axst{(derv)}}
\newcommand{\verax}{\axst{(ver)}}
\newcommand{\veraxb}{\axst{(boxbot)}}
\newcommand{\clax}{\axst{(cl)}}
\newcommand{\claxb}{\axst{(em)}}
\newcommand{\dax}{\axst{(nv)}}
\newcommand{\ndax}{\axst{(nnv)}}
\newcommand{\gdax}{\axst{(gd)}}
\newcommand{\dtax}{\axst{.3}}

\newcommand{\iF}{\mathcal{L}_{\intl}}
\newcommand{\mF}{\mathcal{L}_{\fullsig}}
\newcommand{\iS}{\Sigma}

\newcommand{\acls}{\ensuremath{\mathbf{(a)}\,}}
\newcommand{\bcls}{\ensuremath{\mathbf{(b)}\,}}

\newcommand{\Acls}{\ensuremath{\mathbf{(A)}\,}}
\newcommand{\Bcls}{\ensuremath{\mathbf{(B)}\,}}
\newcommand{\Ccls}{\ensuremath{\mathbf{(C)}\,}}
\newcommand{\Dcls}{\ensuremath{\mathbf{(D)}\,}}
\newcommand{\Ecls}{\ensuremath{\mathbf{(E)}\,}}
\newcommand{\Fcls}{\ensuremath{\mathbf{(F)}\,}}
\newcommand{\Gcls}{\ensuremath{\mathbf{(G)}\,}}

\newcommand{\icc}{\:\unlhd\:}
\newcommand{\cci}{\:\unrhd\:}
\newcommand{\iccs}{\:\lhd\:}
\newcommand{\mcc}{\:\prec\:}

\newcommand{\upset}[1]{\mathtt{Up}_{\icc}(#1)}

\newcommand{\lpl}{\oplus}
\newcommand{\dpl}{ + }

\newcommand{\cmp}{\:;}

\newcommand{\diag}{\Delta}

\newcommand{\reluc}[2]{#1^{#2}\!\!\uparrow}
\newcommand{\reldc}[2]{#1^{#2}\!\!\downarrow}

\newcommand{\ribox}{\boxdot}
\newcommand{\irrbox}{\divideontimes}

\newcommand{\ufpax}{\axst{(ufp)}}

\newcommand{\diagf}[2]{\mathsf{diag}_{#2}#1}

\newcommand{\ilog}{\axst{i}}
\newcommand{\clog}{\axst{cl}}

\newcommand{\iPC}{\axst{IPC}}
\newcommand{\iLC}{\axst{LC}}
\newcommand{\iCl}{\axst{Cl}}
\newcommand{\iSFour}{{\axst{S4}^{\ilog}}}
\newcommand{\cSFour}{\axst{S4}^{\clog}}
\newcommand{\iKFour}{\axst{K4}^{\ilog}}
\newcommand{\cKFour}{\axst{K4}^{\clog}}

\newcommand{\iLax}{{\axst{PLL}^{\ilog}}}
\newcommand{\iCLl}{{\axst{CL}}}
\newcommand{\iK}{{\axst{K}^{\ilog}}}
\newcommand{\cK}{\axst{K}^{\clog}}
\newcommand{\iKW}{{\axst{GL}^{\ilog}}}
\newcommand{\cKW}{{\axst{GL}^{\clog}}}
\newcommand{\iNKW}{\axst{CBL}^{\ilog}}
\newcommand{\iT}{\axst{T}^{\ilog}}
\newcommand{\iCFour}{{\axst{C4}^{\ilog}}}
\newcommand{\iR}{{\axst{R}^{\ilog}}}

\newcommand{\iLob}{{\axst{SL}^{\ilog}}}
\newcommand{\iRLob}{\iLob}
\newcommand{\iSL}{\iLob}

\newcommand{\iNext}{{\axst{CB}^{\ilog}}}
\newcommand{\imHC}{{\axst{mHC}}}
\newcommand{\mHC}{{\axst{mHC}}}

\newcommand{\iKM}{{\axst{KM}}}
\newcommand{\KM}{{\axst{KM}}}
\newcommand{\iNLob}{\iKM}

\newcommand{\iRCl}{{\axst{R}^{\clog}}}
\newcommand{\cR}{{\axst{R}^{\clog}}}

\newcommand{\iD}{\axst{NV}^{\ilog}}
\newcommand{\iND}{\axst{NNV}^{\ilog}}

\newcommand{\iIncn}{\mF}
\newcommand{\iITriv}
{{\axst{Triv}^{\ilog}}}
\newcommand{\iTriv}
{{\axst{Triv}^{\ilog}}}
\newcommand{\iCTriv}{{\axst{Triv}^{\clog}}}
\newcommand{\iIVer}
{{\axst{Ver}^{\ilog}}}
\newcommand{\iVer}
{{\axst{Ver}^{\ilog}}}
\newcommand{\iCVer}{{\axst{Ver}^{\clog}}}

\newcommand{\cwGrz}{{\axst{wGrz}^{\clog}}}
\newcommand{\csGrz}{{\axst{sGrz}^{\clog}}}

\newcommand{\PA}{\ensuremath{\mathsf{PA}}}
\newcommand{\HA}{\ensuremath{\mathsf{HA}}}

\newcommand{\powset}[1]{\mathtt{Pow}(#1)}

\newcommand{\vars}[1]{\overline{#1}}

\newcommand{\qInt}{{\axst{QINT}}}
\newcommand{\qHC}{{\axst{QHC}}}

\newcommand{\HAO}{\ensuremath{\mathsf{HAO}}}
\newcommand{\PABM}{\ensuremath{\mathsf{PABM}}}

\newcommand{\gH}{\mathfrak{H}}

\newcommand{\gT}{\mathfrak{T}}

\newcommand{\hml}{\mathsf{h}}
\newcommand{\iml}{\mathsf{i}}
\newcommand{\jml}{\mathsf{j}}

\newcommand{\wzi}{\axst{[\,i\,]}}
\newcommand{\wzm}{\axst{[m]}}

\newcommand{\wzid}{\axst{\langle\,i\,\rangle}}

\newcommand{\mixax}{\axst{(mix)}}
\newcommand{\cSF}{{\axst{S4}^{\clog}_{\axst{i},\axst{m}}}}
\newcommand{\cSFM}{{\axst{S4M}^{\clog}_{\axst{i},\axst{m}}}}
\newcommand{\cGFM}{{\axst{sGrzM}^{\clog}_{\axst{i},\axst{m}}}}
\newcommand{\biF}{\mathcal{L}_{\axst{i},\axst{m}}}

\newcommand{\wztr}{\flat}

\newcommand{\wzmin}{\tau^{\axst{mix}}}
\newcommand{\wzmax}{\sigma^{\axst{mix}}}
\newcommand{\wzint}{\rho_{\axst{in}}}

\newcommand{\ct}[1]{\mathscr{#1}}

\newcommand{\ctC}{\ct{E}}

\newcommand{\ctE}{\ct{E}}

\newcommand{\ctR}{\ct{R}}

\newcommand{\truem}{\mathtt{true}}
\newcommand{\falsem}{\mathtt{false}}

\newcommand{\Tconn}[1]{#1}
\newcommand{\topT}{\Tconn{\truem}}
\newcommand{\botT}{\Tconn{\falsem}}
\newcommand{\andT}{\Tconn{\wedge}}
\newcommand{\negT}{\Tconn{\neg}}
\newcommand{\orT}{\Tconn{\vee}}
\newcommand{\impT}{\Tconn{\rightarrow}}

\newcommand{\eqcT}{\mathtt{eq}}
\newcommand{\eqT}{\approx}
\newcommand{\leqcT}{\mathtt{leq}}
\newcommand{\inT}{\Tconn{\in}}

\newcommand{\pwoC}[1]{\mathcal{P}{#1}}
\newcommand{\memoC}{\varbigcirc\!\!\!\!\!\!\!\ni\,}
\newcommand{\memmC}{\ni}

\newcommand{\trmo}{\mathbf{1}}
\newcommand{\inio}{\mathbf{0}}
\newcommand{\isom}{\cong}

\newcommand{\subob}[3]{\{ #1 \in #2 \mid #3 \}}

\newcommand{\monar}{\rightarrowtail}
\newcommand{\chrm}[1]{\chi_{#1}}

\newcommand{\dent}[1]{[\!\![#1]\!\!]}

\newcommand{\tmod}[3]{#1,#2 \Vdash #3}
\newcommand{\conss}[3]{#1 \vDash_{#2} #3}
\newcommand{\ofseq}{\overline{f}}
\newcommand{\ocseq}{\overline{c}}

\newcommand{\cpp}{+}
\newcommand{\epim}{\twoheadrightarrow}

\newcommand{\Set}{\mathbf{Set}}

\newcommand{\chr}{\vdash}

\newcommand{\chrmn}{\vdash_{\!\!.}\,}

\newcommand{\obmo}{\odot}
\newcommand{\chobmo}{\synt{\obmo}}

\newcommand{\anymo}{\circledcirc}
\newcommand{\chanymo}{\synt{\circledcirc}}

\newcommand{\finim}{\,\mathtt{fin}
{\scriptstyle{!}}}
\newcommand{\chtim}{\synt{\times}}

\newcommand{\chE}{\synt{E}}
\newcommand{\chF}{\synt{F}}
\newcommand{\chexp}[2]{{#2}^{#1}}
\newcommand{\chtrm}{\synt{\trmo}}
\newcommand{\chom}{\synt{\Omega}}
\newcommand{\chpow}[1]{\synt{\mathcal{P}}#1}
\newcommand{\chto}{\rightsquigarrow}
\newcommand{\chand}{\wedge}

\newcommand{\chtrue}{\synt{\truem}}
\newcommand{\chfalse}{\synt{\falsem}}

\newcommand{\chfin}{\underline{\finim}}

\newcommand{\synt}[1]{\underline{#1}}

\newcommand{\Isubt}{\mathsf{SubTe}}
\newcommand{\Icont}[2]{#1 \subseteq #2}
\newcommand{\Imaxst}{\mathsf{MaxST}}
\newcommand{\Ifix}{\mathsf{fix}\_\mathsf{so}}
\newcommand{\Icontr}{\mathsf{Non}\_\mathsf{exp}}

\newcommand{\currya}[2]{{\mathsf{curry}^{#1}(#2)}}

\newcommand{\eval}{\mathsf{eval}}

\newcommand{\id}{\mathsf{id}}

\newcommand{\prdf}{{\scriptstyle\Pi}_1}

\newcommand{\chprdf}[1]{\underline{{\scriptstyle\Pi}_1}#1}
\newcommand{\prds}{{\scriptstyle\Pi}_2}
\newcommand{\chprds}[1]{\underline{{\scriptstyle\Pi}_2}#1}

\newcommand{\prdar}[1]{\langle#1\rangle}
\newcommand{\chprd}[2]{\underline{\langle}#1,#2\underline{\rangle}}

\newcommand{\cpr}[2]{[#1,#2]}
\newcommand{\cprar}[1]{[#1]}

\newcommand{\chof}[2]{{#1}^{\underline{\wr}}{#2}}

\newcommand{\hm}[2]{\mathsf{hom}^{#2}_{\ct{#1}}}

\newcommand{\ylt}[1]{\breve{#1}}
\newcommand{\ylm}[1]{\hat{#1}}

\newcommand{\point}[1]{(\!\!(\!#1\!)\!\!)}
\newcommand{\pbcl}[2]{[#2.#1]}

\newcommand{\delbirk}{\CIRCLE}

\newcommand{\toptrees}{\mathcal{S}}



%% file: envs_esakiafull.tex
\newtheorem{thm}{Theorem}
\newtheorem{theorem}[thm]{Theorem}
\newtheorem{cor}[thm]{Corollary}
\newtheorem{lem}[thm]{Lemma}
\newtheorem{prop}[thm]{Proposition}
\newtheorem{defn}[thm]{Definition}
\newtheorem{expl}[thm]{Example}
\newtheorem{fact}[thm]{Fact}

\mdfsetup{skipabove=\topskip,skipbelow=\topskip}
\newmdtheoremenv[roundcorner=8pt,backgroundcolor=gray!32,nobreak=true,linecolor=blue]{myremark}[thm]{Remark}

%% file: intKfig.tex






\begin{figure}
\def\jnds{-0.3mm}
\def\njoin{\mathbf{\lpl}}
\def\ndfillclr{blue!10!black!7}
\def\clfillclr{green!10}
\def\mxfillclr{red!5}
\def\isolineclr{black!40}
\def\nlev{3.45}
\def\vlev{14.35}

\begin{tikzpicture}[scale=0.75,xshift=-1.4cm,
nlog/.style={rounded rectangle,draw=black!65,thick,top color=white,bottom color=\ndfillclr,inner sep=0pt,minimum width=1.3cm,minimum height=0.5cm,font=\footnotesize,drop shadow},
botlog/.style={chamfered rectangle,draw=black!50,semithick,top color=white,bottom color=\ndfillclr,inner sep=0pt,minimum width=1.3cm,minimum height=0.45cm,font=\tiny,drop shadow}, 
bignjoin/.style={
inner sep=-0.15mm,
font=\large},
njoin/.style={circle,inner sep=-0.25mm,
font=\scriptsize,fill=white,opacity=1},
intcl/.style={shape=ellipse,font=\tiny,top color=white,bottom color=\clfillclr,inner sep=0,opacity=0.6},
maxconscl/.style={shape=ellipse,font=\tiny,top color=white,bottom color=\mxfillclr,inner sep=0,opacity=0.6},
isoline/.style={dotted,color=\isolineclr,<->},
isotext/.style={font=\tiny,inner sep=0,opacity=0.8}
]


  \node (iSFour) at (9.29, 6.88) [nlog] {$\iSFour$};
  \node (iSFjoin) [below=\jnds of iSFour,njoin] {$\njoin$};

  \node (iKFour) at (5.6, 2.7) [nlog] {$\iKFour$};

  \node (iLax) at (7.67,6) [nlog] {$\iLax$};
  \node (iLaxjoin) [below=\jnds of iLax,njoin] {$\njoin$};

  \node (iK) at (6.78,0.5) [nlog] {$\iK$};

  \node (iKW) at (3.12,4.1) [nlog] {$\iKW$};

  \node (iNKW) at (2.96,6.92) [nlog] {$\iNKW$};
  \node (iNKWjoin) [below=\jnds of iNKW,njoin] {$\njoin$};

  \node (iT) at (10.1,4.88) [nlog] {$\iT$};

  \node (iCFour) at (9.6,3.09) [nlog] {$\iCFour$};

  \node (iR) at (5,4.17) [nlog] {$\iR$};

  \node (iRLob) at (-0.2,7.9) [nlog] {$\iRLob$};
  \node (iRLjoin) [below right=\jnds of iRLob,njoin] {$\njoin$};

  \node (iNext) at (7.6,3.24) [nlog] {$\iNext$};

  \node (imHC) at (6.1,7.95) [nlog] {$\imHC$};
  \node (imHCjoin) [below=\jnds of imHC,njoin] {$\njoin$};

  \node (iKM) at (3.06,11.7) [nlog] {$\iKM$};
  \node (iKMjoin) [below=\jnds of iKM,njoin] {$\njoin$};


  \node (iMLax) at (7.9,10.3) [bignjoin] {$\mathbf{\lpl}$}; 

  \node (iRCl) at (9.48,12.32) [nlog] {$\iRCl$};
  \node (iRCljoin) [below=\jnds of iRCl,njoin] {$\njoin$};

  \node (iD) at (11.73,\nlev) [nlog] {$\iD$};	

  \node (iND) at (0.7,4.3) [nlog] {$\iND$};


  \node (cK) at (14.18,4.6) [nlog] {$\cK$};
  \node (cCFour) at (14.1,6.6) [bignjoin] {$\njoin$};
  \node (cD) at (14.87,6.4) [bignjoin] {$\njoin$}; 
  \node (cKFour) at (12.88,7.8) [bignjoin] {$\njoin$}; 
  \node (cwGrz) at (12.3,9) [nlog] {$\cwGrz$};
  \node (cT) at (14.79,7.8) [bignjoin] {$\njoin$}; 
  \node (cSFour) at (14.68,8.93) [bignjoin] {$\njoin$}; 
  \node (cGrz) at (14.5,11.6) [nlog] {$\csGrz$};
  \node (cGrzjoin) [below=\jnds of cGrz,njoin] {$\njoin$};

  \node (cKW) at (7,12.34) [nlog] {$\cKW$}; 
  \node (cKWjoin) [below=\jnds of cKW,njoin] {$\njoin$};


  \node (iITriv) at (11.65,11.4) [nlog] {$\iITriv$};
  \node (iITjoin) [below left=\jnds of iITriv,njoin] {$\njoin$};

  \node (iCTriv) at (11.9,15) [nlog] {$\iCTriv$};  
  \node (iCTjoin) [below=\jnds of iCTriv,njoin] {$\njoin$};
  \node (iCTRjoin) [below left=\jnds of iCTriv,njoin] {$\njoin$};

  \node (iIVer) at (0,10.75) [nlog] {$\iIVer$};
  \node (iIVjoin) [below right=\jnds of iIVer,njoin] {$\njoin$};

  \node (iCVer) at (1.2,\vlev) [nlog] {$\iCVer$};
  \node (iCVjoin) [below left=\jnds of iCVer,njoin] {$\njoin$};
  \node (iCVNDjoin) [below=\jnds of iCVer,njoin] {$\njoin$};

  \node (iIncn) at (7.0,17) [nlog] {$\iIncn$};
  \node (iIncnND) [below left=\jnds of iIncn,njoin] {$\njoin$};
  \node (iIncnKW) [below=\jnds of iIncn,njoin] {$\njoin$};


  \node [below=0.56cm of iIncn,maxconscl,text width=32mm] {There are no normal logics in this area, $\iCVer$ and $\iCTriv$ are the only maximal consistent ones above $\iK$};


  \path (iK) edge[->] (iD);
  \path (iK) edge[->] (iND);
  \path (iK) edge[->] (iKFour);
  \path (iK) edge[->] (iNext);
  \path (iK) edge[->] (iCFour);
  \path (iND) edge[->] (iRLob);
  \path (iKFour) edge[->] (iKW);
  \path (iKFour) edge[->] (iR);  
  \path (iNext) edge[->] (iT);
  \path (iCFour) edge[->] (iT);
  \path (iD) edge[->] (iT);
  \path (iSFour) edge[->] (iITriv);
  \path (iND) edge[->] (iRLob);
  \path (iMLax) edge[->] (iITriv);  
  \path (iMLax) edge[->] (iRCl);
  \path (iKM) edge[->] (iCVer);	
  \path (iRCl) edge[->] (iCVer);
  \path (iCTriv) edge[->] (iIncn);
  \path (iCVer) edge[->] (iIncn);
  \path (iRLob) edge[->] (iIVer);

  \path (imHC) edge[->] (iKM);

  \path (iNext) edge[->] (cK);
 
  \path (cD) edge[->] (cT);
  \path (cCFour) edge[->] (cT);
  \path (iSFour) edge[->] (cSFour);
  \path (iNKW) edge[->] (cKW);
  \path (cKW) edge[->] (iCVer);
  \path (cwGrz) edge[->] (cKW);
  \path (cwGrz) edge[->] (iRCl);
  \path (cwGrz) edge[->] (iRCl);
  \path (cKFour) edge[->] (cwGrz);
  \path (cGrz) edge[->] node[intcl,sloped,text width=1.5cm] {An isomorphic copy  of $[\iPC,\iCl]$} (iCTriv);


  \draw (iKW) -- (iRLjoin) -- (iR);
  \draw (iKW) -- (iNKWjoin) -- (iNext);
  \draw (iRLob) -- (iKMjoin) -- (iNKW);
  \draw (iR) -- (imHCjoin) -- (iNext);
  \draw (iR) -- (iLaxjoin) -- (iCFour);
  \draw (iKFour) -- (iSFjoin) -- (iT);
  \draw (iR) -- (iITjoin) -- (iT);
  \draw (imHC) -- (iMLax) -- (iLax);
  \draw (iKW) -- (iIVjoin) -- (iCFour);	
  \draw (iIVer) -- node[intcl,sloped,text width=1.2cm] {An isomorphic copy  of $[\iPC,\iCl]$} (iCVjoin) -- (iNext);	
  \draw (iITriv) -- node[intcl,sloped,text width=1.5cm] {An isomorphic copy  of $[\iPC,\iCl]$}(iCTjoin) -- (cK);	
  \draw (iND) -- (iIncnND) -- (iD);	
  \draw (iKW) -- (iIncnKW) -- (iD);	
  \draw (iND) -- (iCVNDjoin) -- (cK);	
  \draw (iR) -- (iRCljoin) -- (cK);	
  \draw (iCFour) -- (cCFour) -- (cK);	
  \draw (iD) -- (cD) -- (cK);	
  \draw (iKFour) -- (cKFour) -- (cK);	
  \draw (iT) -- (cT) -- (cK);	
  \draw (cT) -- (cSFour) -- (cKFour);	
  \draw (iKW) -- (cKWjoin) -- (cK);	
  \draw (iRCl) -- (iCTRjoin) -- (cD);	
  \draw (cwGrz) -- (cGrzjoin) -- (cSFour);

  \path (iKM) edge[isoline] node[isotext]{see \cite{KuznetsovM86:sl}} (cKW);
  \draw[isoline] (imHC) -- 
(cwGrz);

\end{tikzpicture}

\caption{\label{fig:intKfig} \scriptsize Relationships between systems introduced in Table \ref{tab:addaxioms}. Ordinary arrows correspond to inclusion between systems. Lines joined by $\lpl$ denote joins: a join of two logics is of course the smallest normal modal logic containing both. Remember however that \textbf{\textit{not all joins are shown}}, especially in case of joins with weak systems like $\iKFour$ or $\iND$. 
For more on $\cwGrz$, see \cite{Esakia06:jancl} and also \cite{Litak07:bsl} and further references therein. The observation that $\iCTriv$ and $\iCVer$ are the only maximal consistent logics in $\iK$ seems to have been made first by Vakarelov in \cite{Vakarelov81:sl}. Wolter \cite{Wolter97:sl} is one of most advanced studies of the lattice of extensions of $\iK$; in particular, it solves a number of problem posed in \cite{Vakarelov81:sl} and investigates the counterparts of $\Gamma^{\clog}$ above $\iK$ for any given $\Gamma$. Finally, Wolter and Zakharyaschev \cite{WolterZ97:al,WolterZ98:lw} show how to reinterpret this lattice as the lattice of extensions of a certain bimodal classical logic. This is in many ways the most general of Blok-Esakia type results as will be discussed later in Section \ref{sec:complibt}.}

\end{figure}
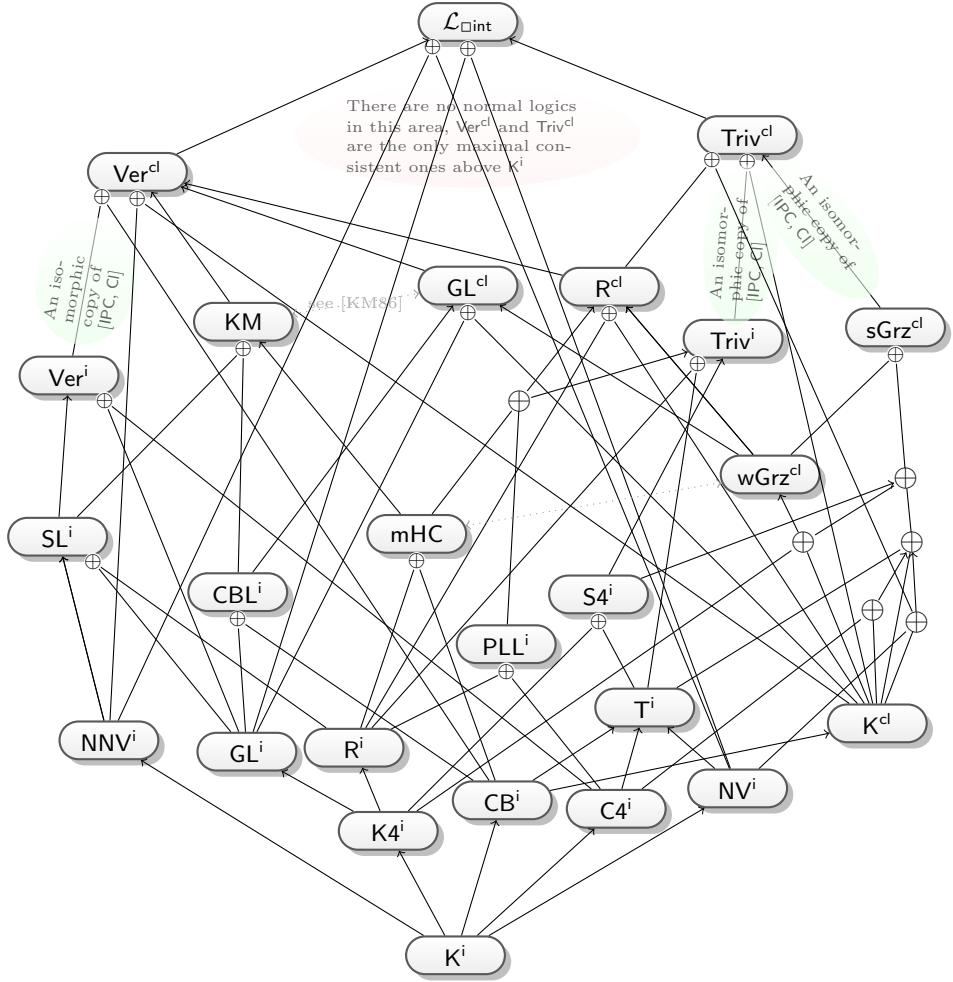

